\documentclass[11pt]{amsart}

\usepackage{amssymb}
\usepackage{amsmath}
\usepackage{amsthm}
\usepackage{amsfonts}

\usepackage{setspace}

\usepackage[T1]{fontenc}
\usepackage[latin1]{inputenc}

\usepackage{graphicx}

\usepackage[top=1.25in, bottom=1.25in, left=1.25in, right=1.25in]{geometry}

\usepackage{chngcntr}
\counterwithin{figure}{section}
\usepackage{bm}
\usepackage{float}

 \usepackage[usenames]{color}
 \definecolor{red}{rgb}{1.0,0.0,0.0}
 
 \definecolor{gre}{rgb}{0.03,0.50,0.03}
 
\definecolor{dvio}{rgb}{0.58, 0.0, 0.83}

\newtheorem{Theorem}{Theorem}[section]
\newtheorem{Proposition}[Theorem]{Proposition}
\newtheorem{Lemma}[Theorem]{Lemma}
\newtheorem{Corollary}[Theorem]{Corollary}
\newtheorem{Definition}[Theorem]{Definition}

\newtheorem{Assumption}[Theorem]{Assumption}
\newtheorem{Remark}[Theorem]{Remark}

\numberwithin{equation}{section}

\def\L{\mathcal{L}}
\def\R{\mathbb R}
\def\N{\mathbb N}
\def\Z{\mathbb Z}

\def\call{{\mathcal L}}

\def\d{\mathrm{d}}

\newenvironment{MyFigure}[1][]{\begin{figure}[#1]\vspace{1.0cm}}{\vspace{1.0cm}\end{figure}}

\begin{document}

\title[]{A dynamic theory of spatial externalities
}

\author[ ]{Raouf Boucekkine$^\MakeLowercase{a}$}
\thanks{$^{a}$Rennes School of Business}
\address{R. Boucekkine, Rennes School of Business, 2 rue Robert d'Arbrissel, Rennes, France.}
 \email{raouf.boucekkine@rennes-sb.com}

\author[ ]{Giorgio Fabbri$^\MakeLowercase{b}$}
\thanks{$^{b}$Univ. Grenoble Alpes, CNRS, INRA, Grenoble INP, GAEL, 38000 Grenoble, France.}
\address{G. Fabbri, Univ. Grenoble Alpes, CNRS, INRA, Grenoble INP, GAEL - CS 40700 - 38058 Grenoble CEDEX 9, France. The work of Giorgio Fabbri is partially supported by the French National Research Agency in the framework of the ``Investissements d'Avenir'' program (ANR-15-IDEX-02) and of the center of excellence LABEX MME-DII (ANR-11-LABX-0023-01).}
\email{giorgio.fabbri@univ-grenoble-alpes.fr}.

\author[ ]{Salvatore Federico$^\MakeLowercase{c}$}
\thanks{$^{c}$Universit\`a degli Studi di Genova, Dipartimento di Economia.}
\address{S. Federico, Universit\`a degli Studi di Genova, Dipartimento di Economia. Via Vivaldi, 5 -- Darsena -- 16126, Italy.} \email{salvatore.federico@unige.it}

\author[ ]{Fausto Gozzi$^\MakeLowercase{d}$}
\thanks{$^{d}$Dipartimento di Economia e Finanza, LUISS \emph{Guido Carli}, Roma.\bigskip}
\address{F. Gozzi, Dipartimento di Economia e Finanza, Libera Universit\`a degli Studi Sociali \emph{Guido Carli}, Roma} \email{fgozzi@luiss.it}

\date{\today }

\begin{abstract}
We characterize the shape of spatial externalities in a continuous time and space differential game with transboundary pollution. We posit a realistic spatiotemporal law of motion for pollution (diffusion and advection), and tackle spatiotemporal non-cooperative (and cooperative) differential games. Precisely, we consider a circle partitioned into several states where a local authority decides autonomously about its investment, production and depollution strategies over time knowing that investment/production generates pollution, and pollution is transboundary. The time horizon is infinite. We allow for a rich set of geographic heterogeneities across states. We solve {\bf analytically} the induced non-cooperative differential game  and  characterize its long-term spatial distributions. In particular, we prove that there exist a Perfect Markov Equilibrium, unique among the class of the affine feedbacks. We further provide with a full exploration of the free riding problem and the associated border effect.

\end{abstract}

\maketitle

\textbf{Key words}:
Spatial externalities, spatial diffusion, differential games in continuous time and space, infinite dimensional optimal control problems, environmental federalism.
\bigskip \noindent

\textbf{JEL classification}:
Q53, 
R12, 
O13, 
C72, 
C61, 
O44. 

\newpage

\linespread{1.35}\selectfont


\section{Introduction}
Spatial externalities are central for the understanding of fundamental phenomena in economics, geography, epidemiology or ecology, among others. Spatial externalities may arise as a result of natural diffusion across space. For example, in the so-called problem of co-existence in ecology (see Barab\'as, 2017, for a modern exposition, and a very nice economic insight into this problem in Beckmann and Wesseler, 2007), farmers who plant a certain type of
crops (e.g. transgenic) may cause negative (or positive) external effects to other farmers who don't, by cross contamination through pollen drift. In the latter example, the externalities do not indeed arise as a result of purposive human action but from natural diffusion. In this paper, we address the problem of (negative) spatial externalities in a context of natural diffusion as in the example above, but in which economic agents/institutions purposively take advantage of natural diffusion to improve their private payoff/welfare, typically through free riding. This covers a large class of applications ranging from pollution control to epidemic control through the more traditional area of natural resources exploitation and the associated tragedy of the commons (Ostrom, 2008). 
In all these areas, free riding in a context of natural diffusion and spatial externalities has a major impact on local and regional spatial patterns.


A very known and documented related area, labeled {\it environmental federalism} by Konisky and Woods (2010), considers transboundary (air, water or soil) pollution. In such a context, the free riding problem has a clear geographic feature: it is associated with a neat border effect. Indeed, if pollution control is decentralized, local governments may be strongly tempted to locate the most polluting facilities near the jurisdictional borders (Monogan III et al., 2017) and/or to enforce less frequently at these borders pro-environmental policies as those promoted by federal or international acts and protocols (see Konisky and Woods, 2010, on the local enforcement of the US federal Clean Air Act).\footnote{There exists a large number of papers documenting that air and water pollution levels are systematically elevated near state borders relative to interior regions. This concerns mainly the US, see Helland and Whitford (2003) or Sigman (2002), but also China, see  Kahn et al. (2015) and Duvivier and Xiong (2013), or Brasil as in Lipscomb and Mobarak (2017).} This paper proposes a general framework which targets the entire class of problems with free riding under natural diffusion and spatial externalities, with a special focus on the environmental problem described just above. We shall formulate the general problem as a differential game in continuous time and space. As natural diffusion involves parabolic differential equations, the involved problems are infinite dimensional. By focusing on the environmental federalism problem, we are able to develop a full analytical approach where the game can be solved in closed form. In particular, we shall prove the existence of a unique Markov Perfect Equilibrium (MPE). We will also  confront the predictions of our game-theoretic model to the abundant empirical literature counterpart.

\smallskip

\smallskip

On the theoretical ground, the free riding problem described above has suggested a quite substantial literature. We refer here for simplicity to the extremely useful works of Hutchinson and Kennedy (2008), and Silva and Caplan (1997).
All these papers use statics frameworks, mostly game-theoretic. To give an example using the same type of geographic space as we will do, Hutchinson and Kennedy (2008) consider a federation of identical states distributed around a circle, each of them occupying an arc of length one. A continuum of identical polluting firms is distributed uniformly along the length of each state, and the mass of firms in each state is normalized to one. Pollution is transboundary as wind is blowing (here from west to east). The authors complete their story of pollution diffusion by assuming an {\it ad hoc} downwind transfer coefficient for emissions per location. In a first stage, the authors show that, under decentralization, states tend to enforce less stringent environmental standards on firms located close to downwind borders, leading  to excessive interstate pollution in equilibrium, which is the standard free riding result. Second, they examine how the interplay between the federal policy on standards and the state policies on enforcement may restore efficiency.

\smallskip

In this paper, while recognizing at it should be the obvious merits of the above mentioned papers, we depart from the latter environmental federalism 
by modeling transboundary pollution as natural diffusion: instead of assuming \emph{ad hoc} pollution diffusion schemes across space, we use a diffusion equation (parabolic partial differential equation) with and without advection. Our choice is not motivated by the mathematical problem involved, which is, as we will see, extremely tricky, but simply by the fact that air or water pollution are hardly {\it ad hoc} and follow precise physical laws. In the case of air pollution, our modelling follows strictly the so-called ``atmospheric diffusion equation'', intensively used in atmospheric physics, see Seinfeld and Pandis, 2016, equation (18.11), from page 768).
Advection allows to introduce non-homogeneous diffusion across space to account for currents or winds for example, which reinforces the realism of our spatiotemporal  modelling. Notice that, since we move to this framework, the resulting games and associated optimal control problems become not only dynamic but also infinite dimensional. A very important parameter in our frame is the so-called diffusion parameter, that is the diffusion speed of pollutants: pollutants with a low speed do not yield the same short-term dynamics as those with higher speeds simply because they do not lead to the same free riding behaviors. This will be made clear in this paper. Only a dynamic framework like ours can highlight this obvious but neglected property.

For simplicity of exposition, we model the space as a circle.\footnote{{The findings of the paper can probably be extended to other models of space as a sphere or other closed surfaces. This would allow to calibrate the model with spatial pollution data and to perform realistic simulations. In this paper we concentrate on the one dimensional model of circle to clearly see and explain the resulting diffusion patterns.}} The circle is partitioned into several states, which need not be identical, contrarily to the above mentioned static theoretical literature. Each state is run by a a local authority which decides about its investment, production and depollution strategies over time knowing that investment/production generates pollution, and pollution is transboundary.\footnote{Note that we allow for heterogeneity also in the discount and risk aversion parameters of the local authorities, heterogeneity across states may also show up in local governance.} The time horizon is infinite. We solve analytically the induced non-cooperative differential game under decentralization, which is itself a far non-trivial task (see the relation to the literature below).

\smallskip

With the closed-form solution paths in hands, we are able to illustrate a series of implications of the model. First of all, we are able to generate the border effect, that's a specific equilibrium behavior near the borders of the states. In particular, we show that net pollution emission flows are increasing as we approach the borders, with strong asymmetries under advection, and that structural breaks show up at the borders. Beside being consistent with the basic theory of spatial externalities, analogous phenomena have been recently disclosed by Lipscomb and Mobarak (2017) in their empirical study on water pollution in Brazil. Second, we are able to detect new regularities and formulate new testable predictions. This is essentially allowed by the fact that our set-up, and the inherent closed-form solution for the MPE, can accommodate a large set of state heterogeneities, contrary to the related  existing theory. For example, our theory yields testable predictions regarding the evolution of the size of spatial externalities (or in other terms, the extent of inefficiency) when the number of states (or jurisdictions) or their size rises. Also the implications of technological, ecological and cultural discrepancy for free riding over space can be derived. 
Last but not least, definitely much easier than the decentralized equilibrium differential game setting, we also characterize the outcomes of the cooperative equilibrium driven by, say, a federal government. This is done to permanently outline the distance to efficiency of the equilibrium counterpart. At the end of the day, we are not only able to find out a closed-form characterization of the spatiotemporal equilibria but also to disentangle a rich set of spatial patterns taking advantage of the many heterogeneities allowed by our methodology (the most elementary being the asymmetry of players). In particular, we show that our model's predictions are consistent with the typical stylized facts highlighted in the related empirical literature.

\smallskip

Beside applicability to other similar problems in other areas as mentioned repeatedly above, our setting can be used to many other environmental issues as important as those raised by environmental federalism. In particular, our frame is perfectly adapted to analyze the currently hot debate around supranational coordination of environmental policies. Indeed, it is general enough to accommodate the two global levels (federal and supranational). Even more important, as our approach allows for deep geographic discrepancies, it is perfectly suitable to study some of the fundamental questions in the international agenda, in particular those related to the North/South environmental divide and the associated debate on the compensations to be given to the South to reach a global deal. This is out of the scope of the current paper.

\subsection*{Relation to the existing literature}

Of course our paper belongs to the literature of dynamic games in the exploitation of the commons (see, for example, Levhari and Mirman, 1980 or Amir, 1996, 2003). The main novelty resides in the infinite-dimensionality of the problem tackled. Indeed, from the methodological point of view, the abstract spatial externality problem we consider is rewritten as a (non-zero sum) differential game in an infinite dimensional state space.
Only very few papers in the mathematical literature (see, e.g., Nisio, 1998 and 1999,
Baltas et al., 2019, or Kocan et al., 1997), deal with such problems which, however, arise in a natural way.

{The paper de Frutos et al. (2021) deserves a specific mention because it is, to the extent of our knowledge, the unique paper dealing with a problem similar to ours, both from an economic and modeling point of view, as they also look at the strategic implications of transboundary pollution in a system driven by a parabolic equation. Even if the treated problem is similar, there are a series of important differences.
First of all, there is a methodological difference: we do not solve HJB equations as in the aforementioned paper, we rather use a direct method adapted to the structure of the functionals, which allows to better understand the structure of the problem.  
 Second, no uniqueness result is  provided by de Frutos et al. (2021), whereas in our work, due to the aforementioned methodology employed,  we are able to prove the uniqueness of open loop equilibria, as well as its uniqueness as a Markovian equilibrium in the class of  in the class of affine feedbacks.  
Third, 
we do prove the convergence of the transitional dynamics to the stationary steady-state, a point not addressed in Frutos et al (2021). Finally, our paper is admittedly more rooted in the economic literature of spatial externalities, with a specific effort to connect with the most recent empirical literature on transboundary pollution (in particular the inherent border effects).}

Related contributions include the papers by de Frutos and Martin-Herran (2019a, 2019b) but in these articles the continuous space-time model ares not studied: analogous discrete-space models are solved delivering feedback Nash equilibria, which are in turn used to capture the spatial interactions among agents through a truly comprehensive set of carefully designed numerical exercises. Except the few above cited papers, the overwhelming part of the differential pollution games studied do not consider spatial diffusion. This is in particular the case of the huge literature on multi-country dynamic games with a common stock of pollution. See Dockner et al (1993), and Boucekkine et al. (2011) for earlier contributions. In the latter literature, pollution dynamics are driven by elementary ordinary differential equations, so that the spatial dimension is simply metaphorical.

\smallskip

In contrast, the (non-strategic) economic literature on spatiotemporal dynamics is rather substantial after the seminal contribution of Brito (2004). In particular, a number of geographic optimal growth models with capital spatiotemporal dynamics have been devised and studied (see Boucekkine et al, 2013, Fabbri, 2016, and Boucekkine et al., 2019, 2021). Another contribution in the same vein but with constant (though space dependent) saving rates is due to Xepapadeas and Yannacopoulos (2016). In all these papers, capital flows across space following a parabolic partial differential equations. Just like the transboundary problems and for the very same reason, the induced problem is infinite-dimensional but there is no strategic ingredient and therefore no game involved. There also exist a few economic papers dealing with transboundary pollution in infinite dimension but without strategic interactions across space, an excellent representative of this stream being due to Camacho and Perez Barahona (2015).
The mathematical literature counterpart is much richer. For optimal control of systems driven by parabolic differential equation (especially in distributed control case) we can mention for instance Lions (1972), Friedman (1982a, 1982b), Barbu (1993), Li and Yong (1995), Lasiecka and Triggiani (2000),  Troltzsch (2000), and Fabbri et al. (2017). But as outlined above, the dynamic games extensions are very scarce.

\bigskip

In the present paper, we show  that, thanks to the special structure of the problem at hand, we are able to express the equilibrium in explicit form. 
Precisely, we rewrite the objective functionals in a transformed way, which allows us to find directly the optimal strategies for the players, and, at the same time, to give a clear economic insight into the properties extracted along the way. The explicit form of the equilibrium also enables us to
characterize it comprehensively and to illustrate some key economic findings readily through complementary numerical exercises. This ultimately shows how effective the machinery of infinite dimensional optimal control can be in studying such type of problems.


\color{black}


This paper is organized as follows. Section 2 proposes a general formulation of the problem. In Section 3 it is specialized in the case of the pollution free riding problem. Section 4 characterizes the cooperative equilibrium. Section 5  and 6 dig deeper in the concept of border effect alluded to above, combining conceptual and numerical analysis, and finally accounting for a rich variety of inter-state heterogeneities. While the previous sections only consider (pollution) diffusion, Section 7 incorporates advection to clearly highlight the implications of non-homogeneous diffusion across space. Section 8 concludes. All the proofs are reported in the Appendix together with a full explanation of the mathematical setting.

\section{A general framework for multi-agents optimization problems with spatial spillovers/externalities}

We consider a dynamic model for a spatial economy subject to spatial spillovers driven by transboundary dynamics.

Even if generalizations are possible\footnote{One can for instance replicate most of the results to the case of the $2$-dimensional sphere surface $S^2$ and indeed many of them can also be generalized to more abstract contexts.} we limit our attention to the case of the circular spatial support $S^1$:
\[
S^1:=\big\{x\in\R^{2}: \ |x|_{\R^{2}}=1\big\}
\]
that is the simplest spatial model being compact and without boundary and then having the significant advantage, in terms of modeling characteristics, of preserving the global dynamics of the involved variables without absorbing or reflecting boundaries (as in the case of possible models on sub-domains with Dirichlet or Neumann boundary conditions). This is also close to the related economic literature, in particular to the geographical setting of Hutchinson and Kennedy (2008).

As usual, in the following we will often describe $S^1$ as the segment $[0, 2\pi]$ with the identification of the two extreme points $0$ and $2\pi$. We denote by $x$ the generic spatial point and with $t\geq 0$ the continuous time coordinate.


\medskip

The global territory $S^1$ is partitioned into a finite number of national states or states/regions within the same country. Each of them is governed by a local public authority (for instance a national, state or regional government) which only takes into account the welfare of the people living on its own territory. More formally we suppose that there are $N$ open intervals in the circle $M_j\subset S^1$, with $j=1,...,N$, such that
$$
M_j\cap M_h=\emptyset \ \ \ \ \mbox{for} \  h\neq j, \ \ h,j=1,...,N.
$$
These regions can cover the whole space $S^1$ --- up to a finite number of boundary points --- but we  also admit the possibility that some residual part with positive measure of the global territory $M_0:=S^1\setminus\bigcup_{j=1}^N M_j$ is uninhabited (for instance seas/oceans).

Each local authority is in charge of economic/political decisions in its own territory. More precisely the authority $j$ chooses the $k$ (control) variables
$$
\bm a_j(t,x)=\left(a_{1,j}(t,x),a_{2,j}(t,x),\dots,a_{k,j}(t,x)\right))
$$
for all $t\in\R_+$ and $x\in M_j$, obeying to suitable constraints so that they take values in a given subset $A$ of $\R^k$. 
We define $\bm a=(a_1,...,a_k):\R_+\times S^1\to A\subseteq \R^k$ as
$$
\bm a(t,x)=\sum_{j=1}^N \bm a_j(t,x)\mathbf{1}_{M_j}(x)+0\cdot \mathbf{1}_{M_0}(x).
$$


We denote by $\bm p=(p_1,...,p_n):\R_+\times S^1\to \R^n$ the state variables common to all the players.
We assume that the dynamics of $\bm p$ are given by the system of controlled semilinear parabolic equations
 $$\begin{cases}
 \frac{\partial p_\ell}{\partial t}(t,x)=\sigma_\ell(x)\frac{\partial^2p_\ell}{\partial x^2}(t,x)+F\left(x,\frac{\partial p_\ell}{\partial x}(t,x),\bm p(t,x),\bm a(t,x)\right),\\
 p_\ell(0,x)=p_{0,\ell}(x),\end{cases}
  \ \ \ \ \ell=1,...,n.$$
where
$F_\ell:S^1\times \R\times \R^n\times A\to\R$
and $\sigma_\ell: S^1\to (0,\infty)$, for $\ell=1,...,n$.

The utility of the local public authority $j$ takes the following form
$$
J_j^{\bm a_{-j}}\big(\bm p_0;\bm a_j)=\int_0^{\infty} e^{-\rho_j t} \left(\int_{M_j} u_j(x,\bm p(t,x), \bm a_j(t,x)) dx\right)dt,
$$
where
$\bm p_0:=(p_{0,1},...,p_{0,n}):S^1\to \R^n$, is the initial spatial distribution of the state variables,
$u_j: M_j\times \R^n\times A\to \R$ is the instantaneous utility of authority $j\in\{1,...,N\}$, and
$\bm a_{-j}$ denotes the control of all the authorities but $j$.
We suppose that the authorities/players engage in a non-cooperative Nash game.

The abstract setting we have presented above can be specified to cover different contexts where space externalities are at work as explained in the introduction section.\footnote{In the case of epidemic dynamics, the state variables are, as in Colombo et al. 2020, the typical ones in the classical behavioral models (like SIR). Pertinent control variables are typically related to lock-down policies (see, e.g., Alvarez et al, 2020) or health enforcement policies (such as tracing, quarantine etc...). The utility, which may differ from authority to authority, should take account of the trade off between effects on economic and humanitarian issues. Another example of specification is the strategic interaction in the exploitation of spatial moving natural resources. For instance in case of fishing 
Territorial Use Rights in Fisheries (TURFs) gives to different actors the right to fish only in a specific zone while the resource moves among the zones. See for example Bressan and Shen (2009). The state variable represents the density of the resource in different locations (which can be driven by pure diffusion/advection dynamic) while the controls are the strength of extraction/harvesting in different sites. Utility can be a function of the extraction/consumption flow in a given zone and can, possibly, include a pro-preservation attitude of the agents.}

Hereafter $p$ will be the pollution spatial-density while the controls are emissions of pollutants in different locations in zones controlled by different authorities. The utility of the planner will take into account the trade-off between producing (and consuming) and disutility coming from pollution.

%
%


\bigskip

The above framework brings to study a differential game in infinite dimensional spaces. A general theory for such problems is missing (only few papers\footnote{See the papers already mentioned by Nisio, 1998 and 1999, Baltas et al., 2019, or Kocan et al., 1997} concerns this topic and none covers our framework. However, using the tools available for the infinite dimensional analysis and control (see e.g. the books Aliprantis and Border, 2006, Li and Yong, 1995,  Fabbri, Gozzi and Swiech, 2017) it is possible to attack specific problems and provide satisfactory solutions by means of infinite dimensional HJBI equations and/or Forward/Backward systems. In specific cases, like the one we are going to present below more specific solution methods can be implemented, giving rise to explicit expressions of the equilibrium strategies.

\section{Pollution externalities: the non-cooperative game}
\label{se:noncoopgame}
We now specialize our abstract framework to a dynamic general equilibrium model for a spatial economy subject to spatial spillovers driven by transboundary pollution dynamics.



As above,  space will be modelled 
by the circular spatial support $S^1$, denoting by $x$ the generic spatial point and with $t\geq 0$ the continuous time coordinate.
As to demographics, we assume the simplest configuration: one individual per location $x$ at any time $t$, so that aggregate and per capita variables coincide at any location.

At any time $t$ and location $x$ the production of the final good $y(t,x)$ depends on the quantity of input used in the production $i(t,x)\geq 0$ according to a linear\footnote{Generalizations including non-linear production functions are possible. In this case the characterization of the equilibrium, together with several qualitative results, remains possible but it's hard to find an explicit solution to the problem.} production function:
\begin{equation}
\label{prod}
y(t,x)= A(x) \; i(t,x),
\end{equation}
where $A(x)> 1$ is the time-independent exogenous productivity at $x$.\footnote{The production input $i(t,x)$ can be interpreted as a capital good but we do not allow for capital accumulation (or equivalently, we assume full depreciation of capital). This is instrumental to obtaining closed-form solutions while preserving the necessary breadth in the analysis of the environmental free riding problem under scrutiny.} The production activity is polluting and a policy of depollution could be advantageously implemented. At any location, output is produced and used for consumption, input and depollution locally. Denoting by $c(t,x)\geq 0$ and $b(t,x)\geq 0$  consumption and  resources devoted to depollution respectively, we have the following resource constraint equation at any location $x$ and time $t$:
\begin{equation}
\label{res}
c(t,x) + i(t,x) + b(t,x)= y(t,x).
\end{equation}
We further assume that using one unit of input produces one unit of emission flow, while a depollution effort $b(t,x) \geq	 0$ can sequestrate a flow $\eta b(t,x)^\theta$, with $\eta\geq 0$ and $\theta\in(0,1)$, of pollutants (we suppose to have decreasing returns in depollution expenditures). Consequently, net emissions are given by
\[
n(t,x) = i(t,x) - \eta (b(t,x))^\theta.
\]
The spatio-temporal dynamics of the pollution stock are subject to two natural phenomena: a diffusion process which tends to disperse the pollutants across the locations, and a location-specific decay $\delta(x)$. All in all, the evolution of the pollution stock $p(t,x)$ is driven by the following parabolic partial differential equation:
\begin{equation}\label{eq:stateequation-sect2}
\begin{cases}
\displaystyle{\frac{\partial p}{\partial t}(t,x) = \sigma \frac{\partial^2 p}{\partial x^2} (t,x)  - \delta(x) p(t, x) + n(t,x),} \\\\
p(0,x)=p_0(x),  \ \ \ x\in S^1,
\end{cases}
\end{equation}
being $\sigma>0$ the diffusivity coefficient measuring the speed of the spatial diffusion of the pollutants and $p_0(x)$ the initial spatial distribution of pollution. We shall introduce advection in Section 6.\footnote{We could have considered a space-dependent diffusion parameter, $\sigma(x)$. Our solution method is unaffected by such a specification just like the many other space-dependent magnitudes  incorporated.}

\medskip

%

Each local authority is in charge of production, consumption and depollution decisions in its own territory and so the authority $j$ chooses $i(t,x)$, $b(t,x)$, $c(t,x)$ for all $t\in\R_+$ and $x\in M_j$ subject to production and resource constraints (\ref{prod}) and (\ref{res}). In order to emphasize that these decisions only concern the region $j$ we will denote them by $i_j(t,x)$, $b_j(t,x)$ and $c_j(t,x)$, so that their relation with the  functions $i(t,x)$ and $b(t,x)$ appearing in (\ref{eq:stateequation-sect2}) are indeed
\begin{equation}
\label{def:ib-sect2}
\big(i(t,x),b(t,x)\big):=
\begin{cases}
\left(i_j(t,
x), b_j(t,x)\right), \ \ \ \mbox{if} \ x\in M_j,\\
0, \ \ \ \ \ \ \ \ \ \ \ \ \ \ \, \ \ \ \ \ \ \ \ \ \mbox{if} \ x\in M_0.
\end{cases}
\end{equation}
We will also denote by $A_j(x)$ the restriction of $A(x)$ to $M_j$.

The utility of the local public authority $j$ depends, as already mentioned, only on the welfare of the population living in its own territory, calculated using the related preference parameters. It takes the following form
\begin{equation}
\label{F-sect2-pre}
\int_0^{\infty} e^{-\rho_j t} \left(\int_{M_j}\left( \frac{\big(c_j(t,x)\big)^{1- \gamma_j}}{1-\gamma_j} - w_j(x)p(t, x) \right) dx\right)dt,
\end{equation}
where $\rho_j>0$ is the discount factor, $\gamma_j\in(0,1)\cup(1,\infty)$ the inverse of the elasticity of intertemporal substitution and $w_j(x)$ is a measure of the unitary location-specific disutility from pollution. The latter can be roughly interpreted as the measure of environmental awareness at location $x$ in state $j$.  It should be noted that we do not assume that all the inhabitants of territory $j$ share the same environmental awareness, which is somehow more realistic. In contrast, the elasticity of substitution is assumed territory-dependent for simplicity.\footnote{Allowing for parameter $\gamma_j$ to depend on location $x$ does not break down the analytical solution.} Finally, each local authority may have a specific view of time discounting, thus the territory-dependent parameter $\rho_j$.


Observe that, even if the expression above only concerns territory $M_j$, it also depends on the choices of other authorities through the variable $p(t,x)$ because, due to pollution diffusion, its value at points $x\in M_j$ depends on all the past production (and thus emission) decisions taken by all the other players (regional authorities). We will make explicit this fact in the notation through the index ``$-j$'', which stands for ``all the index but $j$''. Moreover, supposing that\footnote{This assumption is required by the production function specification (2.1) for the ratio investment to production to be lower than $1$ everywhere and at any time.} $A(x)>1$ at all locations $x\in S^{1}$,
we use (\ref{prod}) and (\ref{res}) to express $c_j(t,x)$ in terms of $i_j(t,x)$ and $b_j(t,x)$ as\footnote{{In the expression $A_j(x)-1$ (and in similar expressions throughout the manuscript) the symbol $1$ represents  the function (the vector) having value 1 at each $x$.}} $(A_j(x)-1) i_j(t, x)-b_j(t,x)$. Finally,  we can write (\ref{F-sect2-pre}) as
\begin{align}
\label{F-sect2}
&J_j^{(i_{-j},b_{-j})}\big(p_0;(i_j,b_j))\nonumber \vspace{.3cm}\\
&:=\int_0^{\infty} e^{-\rho_j t} \left( \int_{M_j}
\frac{\big((A_j(x)-1) i_j(t, x)-b_j(t,x)\big)^{1- \gamma_j}}{1-\gamma_j}
- w_j(x)p(t, x) dx\right)dt.
\end{align}
We now get to formalize how decentralization works in our setting. We suppose that the authorities/players engage in a non-cooperative Nash game. By construction, the
latter is a differential game where each state authority maximizes the spatiotemporal payoff (\ref{F-sect2}) under the state equation (\ref{eq:stateequation-sect2}) subject to positivity constraints on $i_j$, $b_j$ and $c_j$. 

\smallskip


Throughout the text we will assume that the following hypotheses on the regularity of the data and of the parameters are verified: 
\begin{itemize}
	\item[-]
	$p_0\in L^2(S^1;\R_+)$,  $\delta\in C(S^1;\R_+)$, $v\in C^1(S^1;\R)$;
	\item[-] $A_j\in L^\infty(M_j;\R_+)$ and there exists a constant $l$  such that
	$1<l\leq A_j(x)$ for all $j=1,\dots,N$;
	\item[-] For each $j=1,...,N$, one has  $w_j\in C(M_j;\R_+)$ and $w_j$ be extended to a function $\overline{w_j}\in C(\overline{M_j};\R)$ such that $\overline{w_j}(x)>0$ for each $x\in \overline{M_j}$.
\end{itemize}

\color{black}

{Following Section 4.1, in Dockner et al. (2000), we consider two types of strategies: \emph{open loop} and \emph{closed loop Markovian} strategies. After  defining accordingly the two equilibrium concepts, we prove some important existence and uniqueness results. Compared to the scarce related theoretical literature (including de Frutos et al., 2021), these are the main original contributions of this paper. For clarity, we decompose the exposition in two successive subsections starting with the open loop case.}

\subsection{{Open loop equilibria}}
Let, for $j=1,\dots,N$,
\begin{equation}
\label{eq:Cj}
\begin{split}
\mathcal{A}_j:=&\Bigg\{(i_j,b_j): 	\R_+\times M_j\to \R^2_+ \ s.t. \ \  \int_0^{\infty} e^{-\rho_{j}t}\left(\int_{S^{1}} (i(t,x)^{2}+b(t,x)^2)dx\right)dt<\infty\\  & \ \ \ \ \ \ \ \ \ \ \ \ \ \ \ \ \ \ \ \ \ \ \  \mbox{and} \ \  (A_j(x)-1) i_j(t, x)-b_j(t,x)\geq 0 \ \ \ \ \mbox{for all}  \  (t,x)\in \R_+\times M_j\Bigg\}.
\end{split}
\end{equation}

If we are given, for every $j=1,\dots,N$, an element $(i_j,b_j) \in \mathcal{A}_j$, we obtain
a couple $(i,b):\R_+\times S^1 \to\R^2_{+}$ defined as in \eqref{def:ib-sect2}; in this way,
$$(i,b)\in \mathcal{A}:= \mathcal{A}_1\times ... \times \mathcal{A}_N.$$

\begin{Definition}[Open loop admissible strategies]
	The class of \emph{open loop admissible  strategies} is the set $\mathcal{A}.$
\end{Definition}
Observe that, choosing any $(i,b)\in \mathcal{A}$, the state equation \eqref{eq:stateequation-sect2} has a unique solution\footnote{As explained in the Appendix, solution there has to be intended as mild solution to \eqref{SEi} (with $v\equiv 0$), given by \eqref{eq:Ymild}.} for each $p_0:S^1\to\R_+$ square integrable and the functional \eqref{F-sect2} is well defined and finite (see Boucekkine et al. (2021)).

\begin{Definition}\label{def:Nash}
	Let  $p_0:S^1\to\R_+$ square integrable.
	An \emph{open loop (Nash) equilibrium}  for the game starting at $p_0$ is a family of couples $(i_j^*,b_j^*)_{j=1,...,N}\in \mathcal{A}$
	such that, for all $j=1,\dots,N$,
	$$
	J_j^{(i^*_{-j},b^*_{-j})}\big(p_0;(i_j^*,b^*_j) \big)\geq J_j^{(i^*_{-j},b^*_{-j})}\big(p_0;(i_j,b_j)\big),  \ \ \ \ \forall (i_j,b_j)\in \mathcal{A}_j.
	$$
\end{Definition}


As repeatedly mentioned above, we are able to solve analytically for the open loop (Nash) equilibrium involved. This is displayed in the main theorem of our paper below that we are going to present.
{The core of our approach is the rewriting of the functional \eqref{F-sect2}. In Proposition \ref{pr:oo}, it is shown that,
defining $\alpha_j$ as the solution to the following ODE\footnote{{We have to take \eqref{ODEalpha} with $v\equiv 0$; the case $v\neq 0$ is treated in Section \ref{se:advection}, when advection is introduced.}}
\begin{equation}\label{ODEalpha2-sect2bis}
\displaystyle{\rho_j\alpha_j(x)- \sigma \alpha_j'' (x)+\delta(x)\alpha_j(x)=\widehat {w_j}(x), \ \ \ x\in S^1,}
\end{equation}
where \begin{equation}
\label{eq:defwhatj}
\widehat{w_j}(x):=\begin{cases} w_j(x), \ \ \ \mbox{if} \ x\in M_j,\\
0, \ \ \ \ \ \ \ \ \ \mbox{if} \ x\notin M_j,
\end{cases}
\end{equation}
the functional 
$J_j^{(i_{-j},b_{-j})}\big(p_0;(i_j,b_j))$ 
can  be rewritten as 
\begin{footnotesize}
	\begin{align}\label{pppquaterbis}
	&
	J_j^{(i_{-j},b_{-j})}\big(p_0;(i_j,b_j)\big)\\ &=\int_0^{\infty}e^{-\rho_j s}\left(\int_{M_j}
	\left[\frac{\big((A_j(x)-1) i_j(t,x)-b_j(t,x)\big)^{1- \gamma_j}}{1-\gamma_j} - \alpha_j(x) \left(i_j(t,x)-{\eta}b_j(t,x)^{\theta}\right)
	\right]dx\right) ds\nonumber\\
	&-\int_{S^{1}}p_0(x) \alpha_j(x)dx
	-\sum_{k=1, \, k\neq j}^N \int_0^{\infty}e^{-\rho_j  t}\left( \int_{M_{k}}\alpha_j(x) (i_k(t,x)-\eta b_k(t,x)^{\theta})dx \right)dt.\nonumber
	\end{align}
		\end{footnotesize}
		Several properties of the functions $\alpha_{j}$ are described in Appendix \ref{app:alpha}. In particular, Corollary \ref{prop:max} ensures that the objects defined in the theorem below are well posed.}

We introduce here a comparative statics result related to key functions $\alpha_j$ that will be useful in the discussion.
\begin{Proposition}\label{prop:rhodelta}
	$\alpha_j$ is nonincreasing with respect to space-homogeneous increments of $\rho_j+\delta(\cdot)$. and nondecreasing with respect to space-homogeneous increments of $\widehat w_j(\cdot)$.
\end{Proposition}
\begin{proof}
	See Proposition \ref{prop:rhodelta-app} in the Appendix.
\end{proof}

\color{black}	
\medskip

{We will now state our main theoretical results. We start with existence and uniqueness of the open loop equilibrium of the game.}

{\begin{Theorem}
\label{th:oo-sect2}
There exists a unique open loop equilibrium for the described game. It is time-independent\footnote{Since the equilibrium values $i_j^*$ and $b_j^*$ (and consequently of the variables $c_j^*$, $n_j^*$ and $y^{*}_j$ defined below) are time-independent, we will avoid from now to write the time variable in their expressions.}  and it is given, for $j=1,...,N$, by
the expressions 
\begin{align}\label{iopt-sect2}
 i_j^*(x) =   \alpha_j(x)^{-\frac{1}{\gamma_j}}(A_j(x)-1)
^{\frac{1-\gamma_j}{\gamma_j}}
+\left(\eta\theta\right)^{\frac{1}
{1-\theta}}(A_j(x)-1)^{\frac{\theta}
{1-\theta}},
\end{align}
	\begin{equation}\label{bopt-sect2}
b_j^*(x)= \left[(A_j(x)-1)\eta\theta\right]^{\frac{1}{1-\theta}},
\end{equation}
where $\alpha_j$ is the solution of the ODE \eqref{ODEalpha2-sect2bis}.
The associated welfare of player $j$ at the equilibrium is affine in $p_0$:
$$
v_j(p_0)=-\int_{S^1} \alpha_j(x)p_0(x)dx + q_j,
$$
where \begin{align*}
	q_j&:=\frac{1}{\rho_{j}} \int_{M_j}
	\frac{\big((A_j(x)-1) i^*_j(x)-b^*_j(x)\big)^{1-\gamma_j}}{1-\gamma_j}dx
	-\frac{1}{\rho_{j}}\int_{M_j}\alpha_j(x) (i^*_j(x)-{{\eta}} b^*_j(x)^{\theta})dx \\&
	-\frac{1}{\rho_j}\sum_{k=1, \, k\neq j}^N  \int_{M_k}\alpha_j(x)(i^*_k(x)-{{\eta}} b^*_k(x)^{\theta})dx.
	\end{align*} 
	\end{Theorem}}
\begin{proof}
	{See Appendix \ref{app:proofsect3}.}
\end{proof}


In Section \ref{se:geohetero} we will discuss in more detail the dependence of the strategies chosen by the players on the parameters of the model, looking at the effect of different sorts of spatial heterogeneity.
Here we comment briefly on the shape of the players' welfare. One can see (see Remark \ref{rm:deponwanddelta}) that the welfare of player $j$ is an increasing function of pollution rate of decay, $\delta$, and a decreasing function of the pollution disutility parameter, $w_j$. These properties seem obvious if one looks at the objective functionals (for a given set of strategies, it is clearly true that $J_{j}$ has this kind of behavior) but they are \emph{a priori} not straightforward in the context of a Markov (Nash) equilibrium (see again Remark \ref{rm:deponwanddelta}). For instance the effect of increasing $\delta$ is not obvious at first glance because on the one hand it pushes for a direct quicker decay of local and global pollution but, on the other hand, it may lead all players to pollute more (see Proposition \ref{prop:rhodelta}). As for  the effect of $w_j$ on players different from $j$, we observe that decreasing $w_j$ pushes player $j$ to produce and pollute more. Indeed, one can observe that {$i^{*}$} is a decreasing function of $\alpha_j$, which is itself nondecreasing in $w_j$ (see Proposition C.6 in the Appendix for a formal proof). As a result, the initial drop of awareness in territory $j$ can  eventually cause the welfare of other territories to decrease via diffusive pollution.

The equilibrium found has some more specific features, which will be clear in the application sections of this paper. It is important to already single out the role of functions $\alpha_j(x)$. As in the non-game theoretic model in Boucekkine et al. (2021), these functions have a precise meaning. Notice that in the absence of diffusion, one has simply $\alpha_j(x)=\frac{w_j(x)}{\rho_j+\delta(x)}$. In such a case, $\alpha_j(x)$ corresponds simply to the discounted sum of pollution disutilities generated by a unit of pollutant initially located at location $x$. It can be readily shown, as in Boucekkine et al. (2021), that this interpretation is still correct when allowing for diffusion. In sum, $\alpha_j(x)$ represents the total disutility cost over time and space of a unit of pollutant initially located at $x$. This allows to interpret the ME investment decision in a much more appealing way, and will drive many of the findings we will highlight in the application sections.

As a consequence of our theorem, the equilibrium net pollution flow, the production and the consumption in the region $j$ are given respectively by
\begin{equation}\label{netemission}
n_j^*(x) = i_j^*(x) - \eta (b_j^*(x))^\theta,
\end{equation}
\begin{equation}\label{production}
y_j^*(x) = A(x) \; i_j^*(x),
\end{equation}
\begin{equation}\label{consumption}
c_j^*(x)= y_j^*(x) - i_j^*(x)  - b_j^*(x).
\end{equation}


Once we have the set of strategies chosen by the players, we can use them in the state equation (\ref{eq:stateequation-sect2}) to get the corresponding dynamics of the pollution space distribution, ultimately inferring the following asymptotic result.
\begin{Proposition}
\label{pr:pinfty-Nash}
In the context described by Theorem \ref{th:oo-sect2} the spatial  optimal pollution density $p^{*}(t,\cdot)$ converges to the long-run pollution profile $p^{*}_\infty$ which is given by the unique solution of the following ODE:
\[
\sigma p'' (x)  = \delta(x) p(x) - n^*(x), \ \ \ \ x\in S^{1}.
\]
\end{Proposition}
\begin{proof}
See Corollary \ref{cor:Pinfty}.
\end{proof}

We shall use the expressions given above in our numerical exercises in Sections \ref{se:border} and \ref{se:geohetero} where we focus on the spatial distributions of the relevant variable (long-term distribution in the case of pollution).

\subsection{{Closed loop Markov equilibria}} Let
$$\mathcal{A}^{M}_j:=\Big\{(\phi_j,\psi_j):\R_+\times M_j\times L^2(S^1;\R)\to \R_+^2 \ \ \mbox{measurable}\Big\}, \ \ \ \forall j=1,..,N,$$
where $L^2(S^1;\R)$ denotes the space of real valued square integrable functions defined on $S^{1}$,
and let $\mathcal{A}^{M}:=\mathcal{A}^{M}_1\times...\times \mathcal{A}^{M}_N.$
Given $(\phi,\psi)\in\mathcal{A}^{M}$ and $j=1,...,N,$ we denote
$$
(\phi_{-j},\psi_{-j}):= \big((\phi_1,\psi_1), ..., (\phi_{j-1},\psi_{j-1}), (\phi_{j+1},\psi_{j+1}), ..., (\phi_N,\psi_N)\big).
$$

\begin{Definition}\label{adm:markov}
		Let $p_0: S^{1}\to \R_+$ be square integrable.
		An  \emph{admissible Markovian strategy} starting at $p_0$ is an $N$-ple of couples  $(\phi,\psi)=\big((\phi_j,\psi_j)\big)_{j=1,...,N}\in\mathcal{A}^{M}$
		such that the equation\footnote{Notice that (\ref{SEmarkov}) is a PDE with nonlocal terms in the $x$ variable, as $\phi(t,x,\cdot)$ and  $\psi(t,x,\cdot)$ depend, in general, on the structure of $p(t,\cdot)$ on $S^1$, not only on its value at $(t,x)$. The notion of solution we employ is the same as before, i.e., mild solution of the abstract ODE in infinite dimension.}
		\begin{equation}\label{SEmarkov}
		\begin{cases}
		\displaystyle{ p_t(t,x) =  \sigma  p_{xx}(t,x) - \delta(x) p(t, x) +  \phi(t,
			x,p(t,\cdot))-{{\eta}} \psi(t,x,p(t,\cdot))^{\theta},
		}\\\\
		p(0,x)=p_0(x), 
		\end{cases}
		\end{equation}
		where
		$$
		\left(\phi(t,x,p(t,\cdot)), \psi(t,x,p(t,\cdot))\right):=(0,0)\mathbf{1}_{M_0}(x)+\sum_{j=1}^N \mathbf{1}_{M_j}(x) \left(\phi_j(t,
		x,p(t,\cdot)), \psi_j(t,x,p(t,\cdot))\right),$$
		admits a unique solution $p^{(\phi,\psi)}$, and the functional
		$J_j(p_0; (\phi,\psi))$ defined by substituting
		in \eqref{F-sect2}
		$$
		(i_j(t,x),b_j(t,x)):=\Big(\phi_j(t,x,p^{(\phi,\psi)}(t,\cdot), \psi_j(t,x,p^{(\phi,\psi)}(t,\cdot)\Big), \ \ \ p(t,x)=p^{(\phi,\psi)}(t,x),
		$$
		is well defined\footnote{{In particular, for every $(t,x)\in M_j$ and  $j=1,...,N$, the function $$(t,x)\mapsto
		(A_j(x)-1) \phi_j(t,x,p^{(\phi,\psi)}(t,\cdot))-\psi_j(t,x,p^{(\phi,\psi)}(t,\cdot))$$ is nonnegative.}}
 for each $j=1,...,N$.		We denote the class of admissible  Markovian strategies  starting at $p_0$ by $\mathcal{A}^{M}(p_0)$.
\end{Definition}
{We further define the concept of Markovian equilibrium in our context.}
{\begin{Definition}[Markovian equilibrium]\label{def:CLE}
Let  $p_0:S^1\to\R_+$ be square integrable.
		An admissible Markovian strategy $(\phi^*,\psi^*)\in \mathcal{A}^{M}(p_{0})$ is called a \emph{Markovian (Nash) equilibrium} starting at $p_0$ if, for each $j=1,..,N$, it holds
		$$
		J_j(p_0,(\phi^*,\psi^*))\geq J_j\big(p_0; (i_j,b_j), (\phi^*_{-j},\psi^*_{-j})\big),
		$$
		for every $(i_j,b_j)\in\mathcal{A}_j^{\phi^{*}}(p_{0})$, where
		$$
		\mathcal{A}_j^{\phi^{*}}(p_{0})= \{(i_{j},b_{j})\in\mathcal{A}_{j}: \ ((i_j,b_j), (\phi^*_{-j},\psi^*_{-j}))\in\mathcal{A}^{M}(p_0)\},
		$$
		 i.e., if the couple $(i^*_{j},b^*_{j})$ defined by 
		$$i_{j}^{*}(t,x):=\phi^*_{j}(t,x,p^{(\phi^{*},\psi^*)}(t,\cdot)), \ \ \ b_{j}^{*}(t,x):=\psi^*_{j}(t,x,p^{(\phi^{*},\psi^*)}(t,\cdot))$$ 
		is an optimal control for the optimal  control problem
		\begin{small}
	\begin{align*}
	&\sup_{(i_j,b_j)\in\mathcal{A}_j^{\phi^{*}}(p_{0})}\Bigg\{\int_0^{\infty}e^{-\rho_j s}\left(\int_{M_j}
	\left[\frac{\big((A_j(x)-1) i_j(t,x)-b_j(t,x)\big)^{1- \gamma_j}}{1-\gamma_j} - \alpha_j(x) \left(i_j(t,x)-{\eta}b_j(t,x)^{\theta}\right)
	\right]dx\right) ds\\
	&-\int_{S^{1}}p_0(x) \alpha_j(x)dx-
	\sum_{k=1, \, k\neq j}^N \int_0^{\infty}e^{-\rho_j  t}
	\left( \int_{M_{k}}\alpha_j(x) (\phi^{*}_k(t,x,p^{(i_{j}, b_{j})}(t,\cdot))-\eta \psi^{*}_k(t,x,p^{(i_{j}, b_{j})}(t,\cdot))^{\theta})dx \right)dt\Bigg\},
	\end{align*}
		\end{small}
		where $p^{(i_{j}, b_{j})}$ is the solution to 
		\begin{equation*}\label{eq:stateequation-sect2bis}
\begin{cases}
\displaystyle{\frac{\partial p}{\partial t}(t,x) = \sigma \frac{\partial^2 p}{\partial x^2} (t,x)  - \delta(x) p(t, x) +(i_{j}(t,x)-\eta b_{j}(t,x)^{\theta})\mathbf{1}_{M_{j}}(x)}\\ \ \ \ \ \ \ \ \ \ \ \ \ \ \ \displaystyle{+\sum_{k=1, \, k\neq j}^N  (\phi^{*}_k(t,x,p(t,\cdot))-\eta \psi^{*}_k(t,x,p(t,\cdot))^{\theta})\mathbf{1}_{M_{k}}(x),}\medskip\\
p(0,x)=p_0(x),  \ \ \ x\in S^1.
\end{cases}
\end{equation*}
{A Markovian equilibrium is called \emph{stationary} if  the strategies are time independent.}
\end{Definition}}

Since the open loop equilibrium characterized in Theorem \ref{th:oo-sect2} is time-independent, it is also a (trivial) stationary Markovian equilibrium. Given that it does not depend neither on the initial condition $p_0$, it is subgame perfect (see Dockner et al., p. 105). Therefore, it is a Markov Perfect Equilibrium (MPE). We underline this property in the following proposition.
\begin{Proposition}
Suppose that hypotheses of  Theorem \ref{th:oo-sect2} are satisfied. Then the unique open-loop equilibrium characterized in Theorem \ref{th:oo-sect2} is also a stationary MPE.
\end{Proposition}

{{Finally, we study uniqueness of this   equilibrium in the class of affine strategies.
We consider the  case with $\eta =0$ (no depollution) as in de Frutos et al (2021), and we prove uniqueness in the latter class. Accordingly, we  drop the variable $b$ in the functional, as well as in the definition of $\mathcal{A}_j$,  of Markovian strategies, and of Markovian equilibrium.  Let us introduce the spaces
$H=L^{2}(S^{1};\R)$, $H_{j}=L^{2}(M_j;\R)$ and denote by  $H^{+}$ and $H^{+}_{j}$ their positive cones. Moreover, let us denote by 
 $L_{+}(H,H_{j})$ the space of linear bounded positivity preserving --- i.e., mapping nonnegative functions of $H^{+}$ into nonnegative functions of $H^{+}_j$ --- operators from $H$ to $H_{j}$. 
We consider  the following class of stationary affine  Markovian strategies:
\begin{align*}
\mathcal{D} = &\Big\{(\phi_j)_{j=1,...,N}\in\mathcal{A}^{M}(p_{0}) \ \forall p_{0}\in  H^{+}:\\
& \  \ \ \ \phi_{j}(t,x,p)=\phi_{j}(x,p)= (\Phi_{j}p)(x)+\iota_{j}(x),     \ \Phi_{j}\in  L_{+}(H,H_{j}), \ \iota_{j}\in H_{j}^{+}
\Big\}.
\end{align*}
\begin{Theorem}
\label{th:uniqueness}
Let $\eta=0$. The (unique) open loop equilibrium of Theorem \ref{th:oo-sect2} is also the unique {Markovian equilibrium (and then, \emph{a fortiori}, the unique MPE)}  in the class $\mathcal{D}$.
\end{Theorem}}}
\begin{proof}
See Appendix \ref{app:proofsect3}.
\end{proof}
{The theorem above completes our theoretical results in the non-cooperative case. Summing up and given the uniqueness of open loop equilibria established in Theorem \ref{th:oo-sect2},  our analysis provides with: (i) the equivalence between open loop and Markovian equilibria in the class $\mathcal{D}$; (ii)  the uniqueness of Markovian equilibria in the class $\mathcal{D}$. As outlined in the Introduction, these are useful clear-cut contributions to the related mathematical and economic literature, which is particularly thin. We now move to a counterpart cooperative game.}

\section{Pollution externalities: the cooperative game}
\label{se:coopgame}
{Here we investigate the cooperative equilibrium understood as the counterpart social planner problem of the setting above. This will allow us to, as we wrote in the introduction, to measure the distance to efficiency of the non-cooperative equilibrium outcomes. It's worth noting here that this section is a reminiscence of the control non-strategic problem addressed in Boucekkine et al. (2021).}
To design the cooperative solution, we suppose that all the players cooperate to maximize a social welfare function defined as the sum of the utility of all the states/territories of $S^1$:
\begin{equation}
\label{CG-sect3}
\sum_{j=1}^N \int_0^{\infty} \left(\int_{M_j} e^{-\rho_j t} \left( \frac{\big((A_j(x)-1) i_j(t, x)-b_j(t,x)\big)^{1- \gamma_j}}{1-\gamma_j} - w_j(x)p(t, x) \right) dx\right)dt.
\end{equation}
We limit our attention to the case where the preference parameters are the same for all players: $\rho_j=\rho$ and $\gamma_j=\gamma$ for every $j=1,...,N$. In this case the functional (\ref{CG-sect3}) can be rewritten as
\begin{equation}
\label{CG-sect3-bis}
\int_0^{\infty} e^{-\rho t} \left(\int_{S^1}\left( \frac{\big((A(x)-1) i(t, x)-b(t,x)\big)^{1- \gamma}}{1-\gamma} - w(x)p(t, x) \right) dx\right)dt
\end{equation}
where {$w=\sum_{j=1}^{M} \mathbf{1}_{M_{j}} w_{j}$}. It is the standard Benthamite social functional since we suppose that at each location there is exactly one inhabitant.

{The optimal control problem of maximizing (\ref{CG-sect3-bis}) subject to (\ref{eq:stateequation-sect2}) and the positivity constraints on $i$, $b$ and $c$ can be explicitly solved with the same technique; in this case, the core of the solution is represented by the function $\underline{\alpha}$ solution of the ODE
\begin{equation}
\label{ODEalphatris-sect3}
\displaystyle{\rho\underline{\alpha}(x)- \sigma  \underline{\alpha}''(x)+\delta(x)\underline{\alpha}(x)={w}(x), \ \ \ x\in S^1.}
\end{equation}  
The solution is  described in the following theorem, which proof proceeds along the same line as the one of Theorem \ref{th:oo-sect2}\footnote{{Notice that we deal in this case with just a Markovian optimal control problem. So, considering open loop or closed loop Markovian strategies is equivalent (see Dockner et al. (2000), Ch.\,4)}.}.}
\begin{Theorem}
\label{tt-sect3}
{There exists a unique  equilibrium for the described cooperation game. It is  time-independent and is  given by
\begin{align}\label{iopt3-sect3}
 \underline i^*(x)
&= \left[\underline{\alpha}(x)^{-\frac{1}{\gamma}}(A(x)-1)
^{\frac{1-\gamma}{\gamma}}
+\left(\eta\theta\right)^{\frac{1}
{1-\theta}}(A(x)-1)^{\frac{\theta}
{1-\theta}}\right]
\end{align}
\begin{equation}\label{bopt3-sect3}
 \underline b^*(x)= \big[(A(x)-1)\eta\theta\big]^{\frac{1}{1-\theta}},
\end{equation}
where
$\underline{\alpha}$ is the solution of
\eqref{ODEalphatris-sect3}.
The corresponding welfare is
$$
\underline{v}(p_0):=-\int_{S^1}\underline{\alpha}(x)p_0(x)dx + \underline{q},
$$
where
\begin{align*}
	\underline{q}&:=\frac{1}{\rho} \int_{S^1}\left[
	\frac{\big(({A}(x)-1) \underline{i}^*(x)-\underline{b}^*(x)\big)^{1-\gamma}}{1-\gamma}- \underline{\alpha}(x) (\underline{i}^*(x)-{\eta} \underline{b}^*(x)^{\theta})\right] dx.
	\end{align*}
}
\end{Theorem}
{As for the Nash equilibrium case, the variables  $\underline{n}^*$, $\underline{c}^*$, and $\underline{y}^*$ are similarly defined.}
A counterpart of the asymptotic result given in Proposition \ref{pr:pinfty-Nash} is given, in the cooperative context, by the following proposition.
\begin{Proposition}
\label{pr:pinfty-coop}
In the context described by Theorem \ref{tt-sect3} the {optimal} spatial pollution density {$\underline p^{*}(t,\cdot)$} converges to the long-run pollution profile {$\underline{p}^{*}_\infty$} which is given by the unique solution of the following elliptic equation:
\[
\sigma \underline{p}''(x)  = \delta(x) \underline{p}(x) - \underline{n}^*(x),
\]
where
$\underline{n}^*(x) = \underline{i}^*(x) - \eta (\underline{b}^*(x))^\theta $ (and $ \underline{b}^*$ and $\underline{i}^*$ are defined in Theorem \ref{tt-sect3}).
\end{Proposition}
Not surprisingly, the Nash equilibrium described in Theorem \ref{th:oo-sect2} is markedly different from the social optimum described above and it is suboptimal in terms of the chosen social target. The suboptimality is of course driven by the spatial pollution externality: the local authority does not completely internalize the damages of the emissions produced in its territory, part of it leaves
its territory and harms instead the welfare of other territories, especially the closest. Interestingly enough, our comprehensive modelling of pollution diffusion allows for an accurate appraisal of this externality. A crucial parameter in this respect is parameter $\sigma$: when $\sigma$ diminishes, the diffusion is slower. When $\sigma$ is equal to $0$, the spatial dynamics vanish. In this case the model is spatially degenerate in the sense that there is no interaction among the regional economies, resulting in the degenerate pointwise maximization of the functional
\[
\int_0^{\infty} e^{-\rho t} \left( \frac{\big((A(x)-1) i(t, x)-b(t,x)\big)^{1- \gamma}}{1-\gamma} - w(x)p(t, x) \right) dt
\]
for any fixed $x\in S^1$ subject to
\[
\displaystyle{\frac{\partial p}{\partial t}(t,x) =  - \delta(x) p(t, x) +  i(t, x)-\eta b(t,x)^{\theta}}.
\]
The solution of this maximization problem is given by the same value of $b$ given in (\ref{bopt-sect2}) and
\[
i(x)= \left(\frac{w(x)}{\rho+\delta(x)}\right)^{-\frac{1}{\gamma}}(A(x)-1) ^{\frac{1-\gamma}{\gamma}} +\left(\eta\theta\right)^{\frac{1}
{1-\theta}}(A(x)-1)^{\frac{\theta}{1-\theta}}.
\]
In such a case, since pollution remains in the territories where it is produced, no spatial externality arises. Accordingly, the cooperative  and the non-cooperative equilibria should coincide when $\sigma=0$.
Actually Theorem \ref{tt-sect3} tells us a bit more: the no-diffusion non-cooperative case is conceptually close to the cooperative case with any given diffusion speed. Indeed, it can be shown that the two coincide in the particular case where $\delta$ and $w$ are constant, in terms of chosen $i$, $b$ and $c$. This intriguing property is disclosed in the proposition below.	

\begin{Proposition}
\label{pr:twobenchmarks}
Suppose that $\delta$ and $w$ are constant in space.
Then the strategies $b_{j,0}^*(x)$ and $i_{j,0}^*(x)$ of player $j$ in the Nash equilibrium described in Theorem \ref{th:oo-sect2} when $\sigma=0$ are the same as the equilibrium strategies $\underline{b}_j^*(x)$ and $\underline{i}_j^*(x)$ of the cooperative case described in Theorem \ref{tt-sect3}, for every $\sigma\ge 0$.
\end{Proposition}
\begin{proof}
See Proposition \ref{cor:duebanchmark} in the Appendix.
\end{proof}

\section{Border Effects}
\label{se:border}

We now perform a series of numerical exercises to uncover the main features of the long-term equilibrium spatial distributions. This will allow us to visualize at glance the shape of the border effects associated with the free riding problem under decentralization. Essentially, we will make clear that the shapes generated are quite consistent with the predictions of the basic (static) theory of spatial externalities, particularly in what concerns the specific equilibrium behavior near the borders of the states: net pollution flows will be shown to be larger as the borders are approached, and  structural breaks emerge at the borders. In this section, we consider that all the states are geographically identical in terms of various parameters of the model, except possibly for the size while the next section will dig deeper in the implications of other spatial heterogeneities.

We shall report at the same time the induced distributions for the cooperative game, which will make the border effects even more striking. We start by showing the situation of two territories while, in a second step, we theoretically address another implication of the basic theory of spatial externalities, that is the evolution of the size of spatial externalities (or in other terms, the extent of inefficiency) when the number of states (or jurisdictions) rises.

\subsection{The shape of border effects}
We concentrate here our attention on cases where all the parameters are constant in space. Accordingly, the unique geographic discrepancies result from the partition in decentralized territories (in the Nash case), that is in the existence of borders and possible differences in territory size.


We will compare the Nash equilibrium described in Theorem \ref{th:oo-sect2} with the benchmark described in Section \ref{se:coopgame}.
In particular, since, as announced above,  $\delta$ and $w$ are constant in space, we are always in the context of Proposition \ref{pr:twobenchmarks} so that the cooperative benchmark with diffusion of Theorem \ref{tt-sect3} and the non-diffusion benchmark are always equivalent. To lighten the notation we will omit the index $j$.


\smallskip

\subsubsection{Symmetric two-region cases}
We start with the case where territories have equal size. In the common economic language, this amounts to studying the symmetric two-region case.

\begin{MyFigure}[H]
\centering
\includegraphics[width=\linewidth]{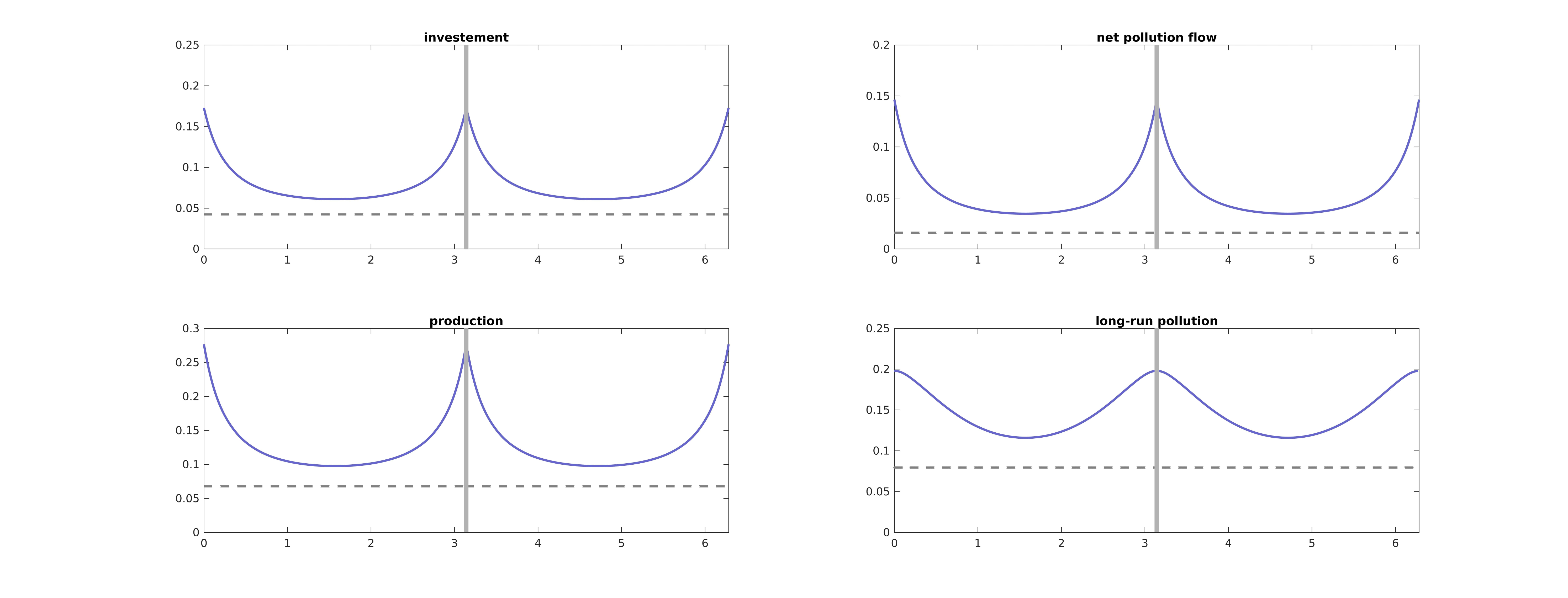}
\caption{\footnotesize
The case of two players each controlling one half of the circle. All the parameters are space independent: the productivity factor $A$ is $1.6$, the discount $\rho$ is $0.03$, the diffusivity  $\sigma$ is $0.5$, the natural decay of the pollution $\delta$ is $0.2$, the weight of pollution in the utility parameter $w$ is $1$, the intertemporal substitution parameter $\gamma$ is $0.5$, the efficiency of depollution expenditure factor $\phi$ is $0.2$ while the scale returns of depollution activity $\theta$ is $0.4$. The continuous (and blue) lines represent the Nash equilibrium while the dashed (and gray) lines reproduce the cooperative benchmark or, equivalently, the non-diffusion benchmark.}
\label{fig:2pl-tuttoconst-sect4}
\end{MyFigure}

In Figure \ref{fig:2pl-tuttoconst-sect4} each local authority controls one half of the circle and then it is interested only in utility of its region (parameters' values are given in the caption of the figure). The continuous (and blue) lines represent the values of the variables in various locations at the Nash equilibrium while the dashed (and gray) lines reproduce the cooperative benchmark or equivalently the non-diffusion benchmark. We represent four variables: the optimal investment $i^*$, the optimal net pollution $n^*$, the optimal production $y^*$,  and the long-run optimal pollution profile $p^*_\infty$. The latter is obtained in particular thanks to the representation of the solution in series given in Appendix \ref{secapp:Fouries}. We do not represent in the figure the depollution effort at the equilibrium since, given the particularly simple situation (all the parameters are constant in space), it is constant over space. For the same reasons we can also observe that the qualitative behavior of net emissions is exactly the same as that of investment and of production.

The first evident element in the distribution of investment (and then in those of production and net pollution) are the big differences among locations of a same region. In particular we can observe that investment and economic activity are particularly strong near the borders of each region. This \emph{border effect} is due to the spatial structure of the externalities: the negative effects of the emissions on the utility are less and less internalized by the local authority as the location gets closer to the border. That's because a greater part of the pollutants in these locations will flow into another territory and has therefore to be managed by another local authority. In the symmetric two-region case with positive diffusivity that we have here, the emission at the boundary points are immediately equally shared by the two territories while the pollutants coming from a far interior point remain in the short run mostly in their ``native'' region. The source of the externality inefficiency can be well visualized looking at the long-run distribution of  pollution: indeed the concentration of pollution at the boundary is much less pronounced than the corresponding peak of input because a significant part of the pollutants leaves the locations where they are originally produced.

\begin{MyFigure}[H]
\centering

\includegraphics[width=\linewidth]{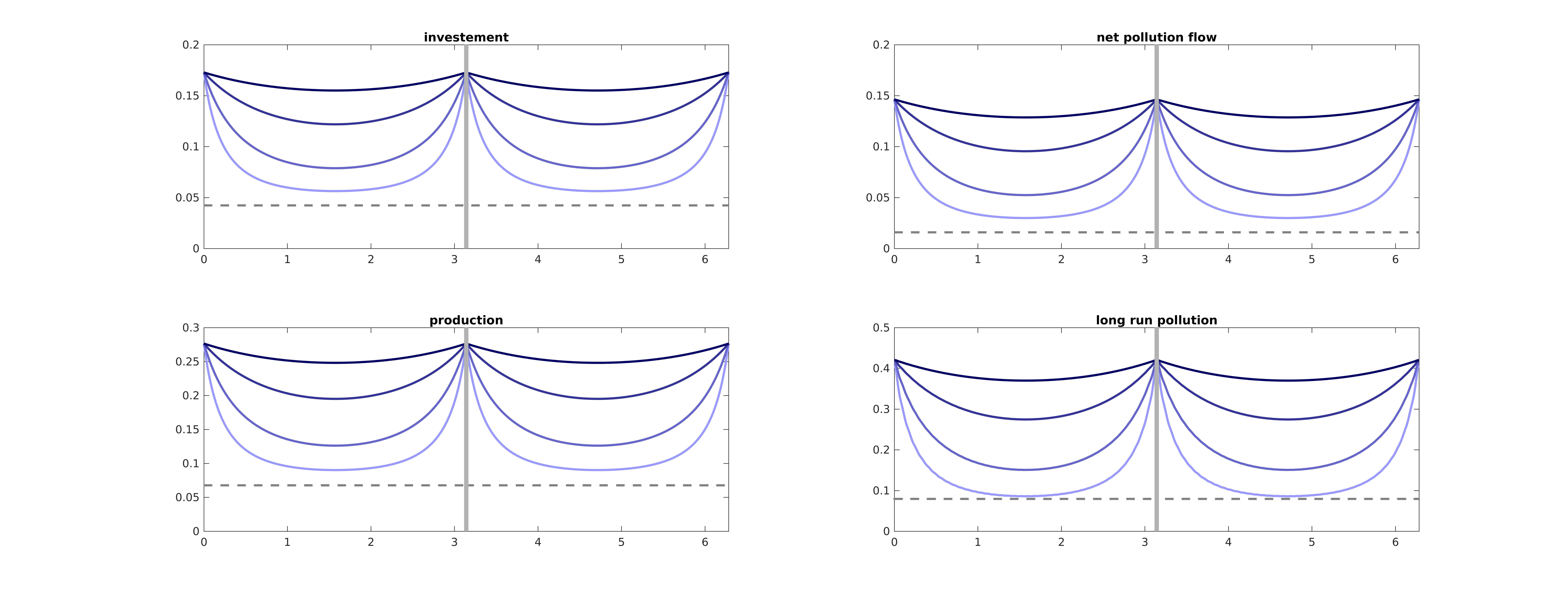}
\caption{\footnotesize
The case of two players controlling one half of the circle varying the diffusivity coefficient. The values of the parameters (all constant over the whole space) are the same as in Figure \ref{fig:2pl-tuttoconst-sect4} except the values of diffusivity $\sigma$ which takes now four values: $0.4$ (the lightest line), $0.8$, $1.6$, and $3.2$ (the darkest line). Continuous (and colored) lines represent the Nash equilibrium while the dashed (and gray) lines are the cooperative benchmark or, equivalently, the non-diffusion benchmark.}
\label{fig:2pl-variarediff-sect4}
\end{MyFigure}

In Figure \ref{fig:2pl-variarediff-sect4} we dig a little deeper in the mechanisms at work and we see what happens when we vary the coefficient $\sigma$, the pollution diffusion speed. We have again the same symmetric two-countries situation as in Figure \ref{fig:2pl-tuttoconst-sect4} and we represent with a colored continuous (respectively, gray dashed) line the variables at the Nash equilibrium for various values of the parameter $\sigma$ (respectively, at the benchmark). We choose four possible values for $\sigma$: $0.4$ (the lightest line), $0.8$, $1.6$, and $3.2$ (the darkest line), all the other parameters are the same as in Figure \ref{fig:2pl-tuttoconst-sect4} (so the variables' values at the benchmark are the same).


Not surprising, the lower $\sigma$, the more the system tends to mimic the behaviour of the $0$-diffusion benchmark.
Conversely, the higher the value of $\sigma$, the faster pollution disseminates across locations. For $\sigma$ big enough, the situation in each location is similar to the situation we have at the boundary since after a short while the generated emissions have an almost equal probability of being in the territory of both authorities. In other words, the higher the value of $\sigma$, the lesser \emph{future} negative effects on the utility are internalized by the local authority and, thus, the higher the chosen level of input, production and emissions in internal points. This mechanism highlights the intertemporal role of the parameter $\sigma$ that will be emphasized even more just below.


The two limits of the equilibrium profile of $i_j^*$ when $\sigma\to 0^+$ and $\sigma\to +\infty$ can be computed explicitly and they are (see Proposition \ref{pr:limitsigmato0toinfty} in the Appendix), for the general case specified in Theorem \ref{th:oo-sect2},
\begin{equation}\label{io}
i_j^{*,0}(x)
= \left ( \frac{w_j(x)}{\rho_j+\delta(x)} \right )^{-\frac{1}{\gamma_j}}(A_j(x)-1)
^{\frac{1-\gamma_j}{\gamma_j}}
+\left(\eta\theta\right)^{\frac{1}
{1-\theta}}(A_j(x)-1)^{\frac{\theta}
{1-\theta}},
\end{equation}
and
\begin{equation}\label{iinfty}
i_j^{*,\infty}(x)
= \left ( \frac{\int_{S^1}w_j(x)dx}{\int_{S^1}(\rho_j+\delta(x))dx} \right )^{-\frac{1}{\gamma_j}}(A_j(x)-1)
^{\frac{1-\gamma_j}{\gamma_j}}
+\left(\eta\theta\right)^{\frac{1}
{1-\theta}}(A_j(x)-1)^{\frac{\theta}
{1-\theta}}.
\end{equation}
Notice that
the latter expression depends on the space location $x$ only through $A(x)$. In the  case of 2 symmetric agents and spatial constants parameters, one gets the two following spatial-independent expressions
\begin{equation}\label{iobis}
i^{*,0} = \left ( \frac{w}{\rho+\delta} \right )^{-\frac{1}{\gamma}}(A-1)^{\frac{1-\gamma}{\gamma}} +\left(\eta\theta\right)^{\frac{1} {1-\theta}}(A-1)^{\frac{\theta} {1-\theta}},
\end{equation}
and
\begin{equation}\label{iinftybis}
i^{*,\infty}= \left ( \frac{1}{2} \frac{w}{\rho+\delta} \right )^{-\frac{1}{\gamma}}(A-1)
^{\frac{1-\gamma}{\gamma}}
+\left(\eta\theta\right)^{\frac{1}
{1-\theta}}(A-1)^{\frac{\theta}
{1-\theta}}.
\end{equation}

One can readily check that consistently with the intepretations given above, $i^{*,\infty}> i^{*,0}$.


\smallskip
\subsubsection{Non-symmetric two-region cases}
We now move to asymmetric two-region cases. In Figure \ref{fig:2pl-unquarto3quarti-sect4} we represent a situation where the parameters are  space-independent. Thus, heterogenous behaviours across space (if any) are entirely due the existence of borders (and to the subsequent differences in the way agents value utility of people living in different locations). All the parameters are the same as in Figure \ref{fig:2pl-tuttoconst-sect4} but here the dimensions of the regions governed by the two authorities are different: the first controls three fourths of the circle while the second only governs on a fourth of the space. The effect of the new division of the territory relative to the situation of Figure \ref{fig:2pl-tuttoconst-sect4} is neat: the authority with the larger territory internalizes more the effect of its emissions because it anticipates that it will have a substantial part of them back to its territory in the future. For the very opposite reason, the authority which has a fourth of the circle is less affected by its own emissions, leading to production and pollution levels larger than in Figure \ref{fig:2pl-tuttoconst-sect4} (and \emph{a fortiori} than the player controlling three fourths of the circle). In other words, a major implication of our ``circular" and geographically homogeneous world is: the bigger the country, the cleaner! This is a quite interesting result. Of course, by construction, if intrinsic geographic heterogeneity is added (in technology or ecology for example), the latter result might be reversed. But it is important to visualize the benchmark result with only differences in size and no further geographic discrepancy.

\begin{MyFigure}[H]
\centering
\includegraphics[width=\linewidth]{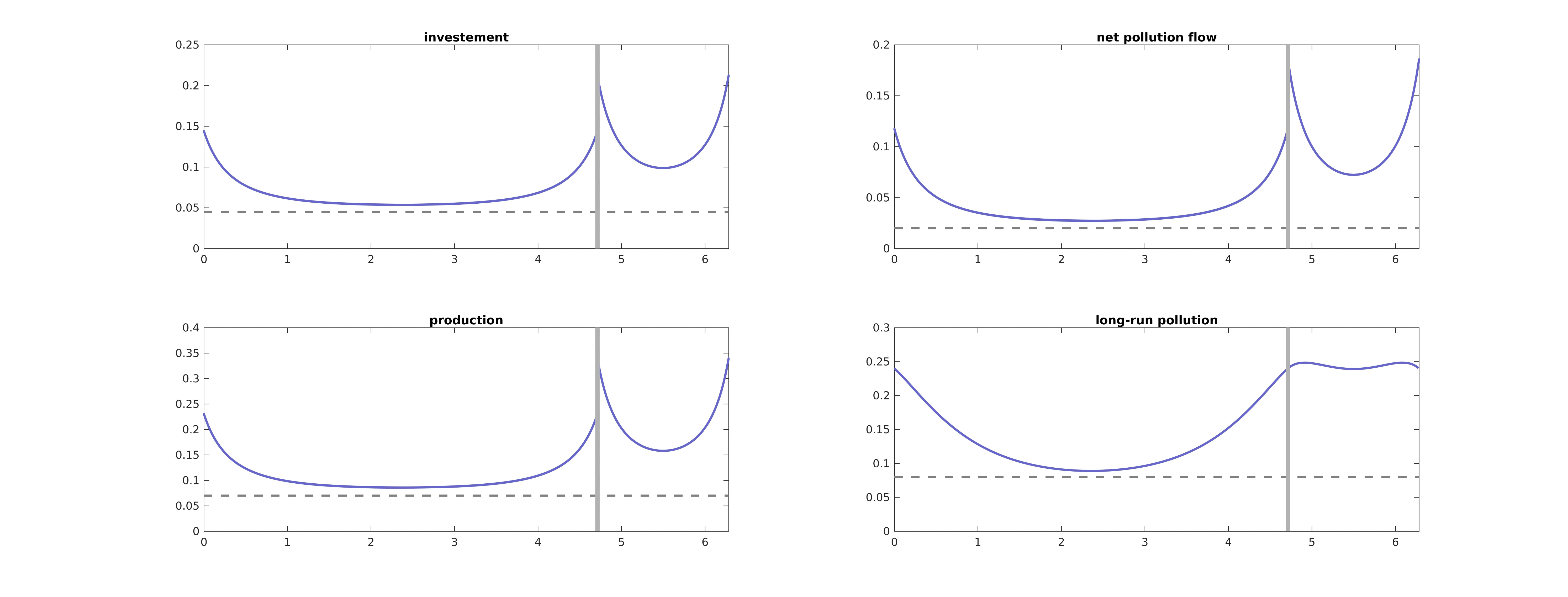}
\caption{\footnotesize
The case of two players controlling respectively one fourth and three fourths of the circle. The values of the parameters (all constant over the whole space) are the same as in Figure \ref{fig:2pl-tuttoconst-sect4}. Continuous (and colored) lines represent the Nash equilibrium while the dashed (and gray) lines are the cooperative benchmark or, equivalently, the non-diffusion benchmark.}
\label{fig:2pl-unquarto3quarti-sect4}
\end{MyFigure}


\begin{Remark}
The strength and the structure of the border effects we mentioned also depend on the model of the spatial structure we use. A structure like $S^1$ allows us to consider the global effects of pollution in the sense that, as already pointed out, the agents (especially those that govern larger areas) are led to consider the fact that part of the emissions that leave the territory, contribute to increase a spatial stock of pollution that will return to the country in the future. Given the compactness of the support it can also happen that the pollution ``crosses'' the entire territory of the other player before returning. In other types of models the situation may be different.  For example if one considers the spatial model of the segment with absorbing borders, the pollution leaving the area controlled by a player will no longer return to it. If instead one considers a model of infinite space, the effect of compactness disappears and therefore the difference of behaviors between the player controlling a large territory and the one controlling a small territory is reduced.
\end{Remark}

\subsection{Number of players and the size of externalities}

Exactly the same kind of behavior we have seen in the last numerical exercises arises from increasing the number of players: when several players 'more than $2$) with a per-capita small territory interact, each of them internalizes only a modest portion of the damages caused by their territory's emissions, which leads each of them to increases input, production and pollution with respect to the two players case. Therefore the total amount of pollution tends to increase with the number of players. We devote to this finding a deeper theoretical foundation here.

We consider again the case where the coefficients of the problem are homogeneous in space. 
Among the sets of territories configurations  described in Section \ref{se:noncoopgame} (see also Appendix \ref{appsec:formulation}), we introduce the following partial order relation: given two configurations
$$
\Pi^1=\{M_1^1,...,M_N^1\}, \ \ \ \ \Pi^2=\{M_1^2,...,M_K^2\}
$$
we say that $\Pi^1\preceq  \Pi^2$ if
$$
\forall j=1,...,N, \ \ \exists i=1,...K: \ \ \  M^1_j\subseteq M_i^2.
$$
This means that the territories configuration $\Pi^1$ is a fragmentation of the territories configuration $\Pi^2$.
\begin{Proposition}
	\label{prop:comparison}
	Let $\Pi^1, \Pi^2$ be two territories' configurations such that $\Pi^1\preceq \Pi^2$ and let $p^{1,*},p^{2,*}$ the associated optimal pollution paths. Then $p^{1,*}(t,x)\geq p^{2,*}(t,x)$ for all $(t,x)\in \R_+\times S^1$.
\end{Proposition}
\begin{proof}
	See Appendix \ref{sub:terr}.
\end{proof}

The result displayed by the proposition is sharp: when decisions in terms of investment, production and emissions of a certain territory move from a central entity to smaller sub-entities, the pollution levels systematically increase in line with the resulting impact on the degree of internalization of externalities by each player. Again, this broadens the counterpart property in the basic static theory of spatial externalities. It is also (somehow by construction) consistent with the intuitions one can gain from the numerical exploration above. Note that the result is obtained under the assumption of geographic homogeneity. The next section is designed to highlight some the implications of adding geographic discrepancies into the theory, a feature not considered generally in the standard theory (see for example, Hutchinson and Kennedy, 2008). Our analytical approach allows for these discrepancies as our closed-form solutions do encompass them.

\section{Geography and Heterogeneity}
\label{se:geohetero}

We look now at what happens when we introduce natural, technological or preference differences among the regions. We stick to the symmetric two-region case since
the relevant mechanisms are already at work there. To be able to disentangle various effects we will look at the effect of one parameter each time. We start with technology.

\subsection{Geographic discrepancy in technology}

\begin{MyFigure}[H]
	\centering
	\includegraphics[width=\linewidth]{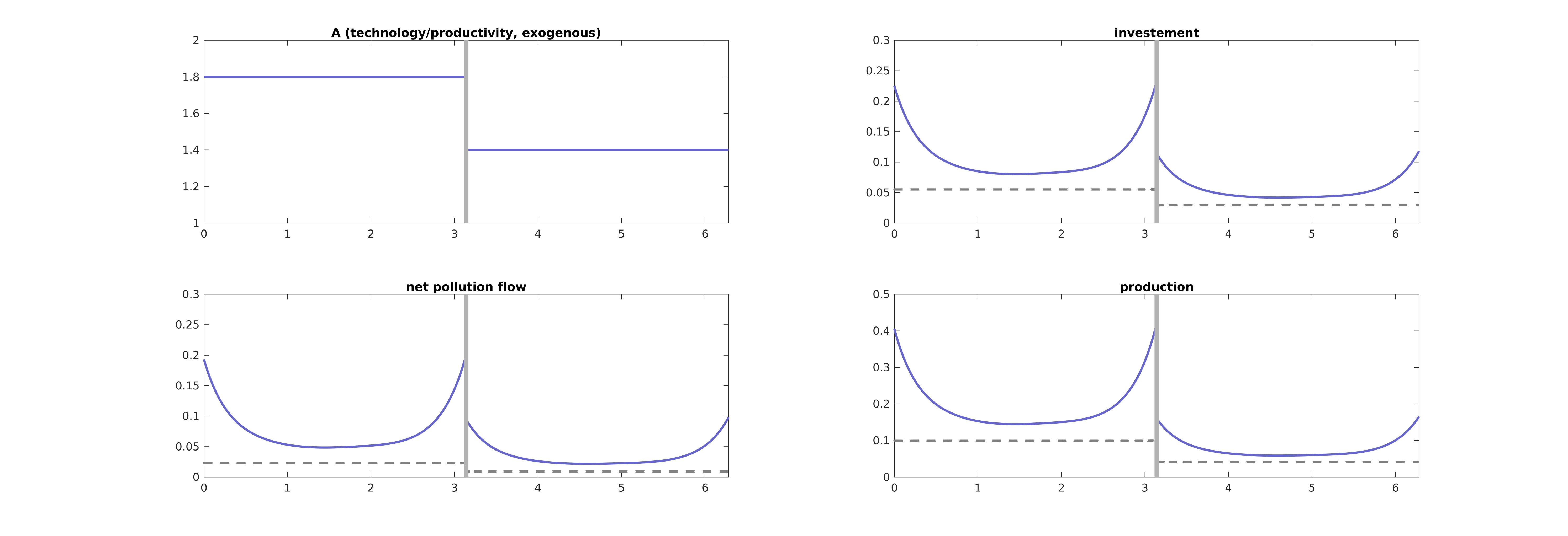}
	\caption{\footnotesize
		The case of two players controlling one half of the circle with different levels of technology $A$. The productivity of the first player (on the left) is $A_1$ (constant it her territory) equal to $1.8$ while the value of $A_2$ is $1.4$. All other parameters (are constant over the space and) are the same as in Figure \ref{fig:2pl-tuttoconst-sect4}. Continuous (and colored) lines represent the profiles of the variables at the Nash equilibrium while the dashed (and gray) lines are related to the cooperative benchmark or, equivalently, the non-diffusion benchmark.}
\label{fig:heterogeneousA}
\end{MyFigure}

In Figure \ref{fig:heterogeneousA} we consider an economy with two regions at different technological levels, that's with different productivities. As we can see from (\ref{iopt-sect2}), the effect on investment of varying $A$ depends on the value of other parameters (in particular on $\gamma$, see Remark \ref{rm:netemissions} in the Appendix for technical details). Indeed a variation of $A$ produces a typical income-substitution trade-off: on the one hand increasing $A$ makes investment more productive leading to a higher  increase in investment relative to consumption and depollution effort. On the other hand, a higher level of $A$ can guarantee a higher level of production and thus of consumption together with lower investment and subsequently lower emissions. The predominance of one of the two channels depends on the values of  various parameters. In Figure \ref{fig:heterogeneousA} the first effect is stronger. Conversely, as one can see from the expression of $b_j$ at the equilibrium given in (\ref{bopt-sect2}), the impact of increasing $A$ on the depollution expenditures is always positive: increasing productivity always gives more room for pro-environmental actions.

The effects on net pollution depend on the relative strength of the mechanisms described above. For our selected parametrization, the outcomes are represented in Figure \ref{fig:heterogeneousA}. The long-run spatial distribution of pollution is not explicitly represented in the figure but, similarly to what we have in the figures of Section \ref{se:border} it is a smoothed version of the distribution of net emission flows (see Remark \ref{rm:depPinfty} for a more technical observation on that).

\begin{MyFigure}[H]
	\centering
	\includegraphics[width=\linewidth]{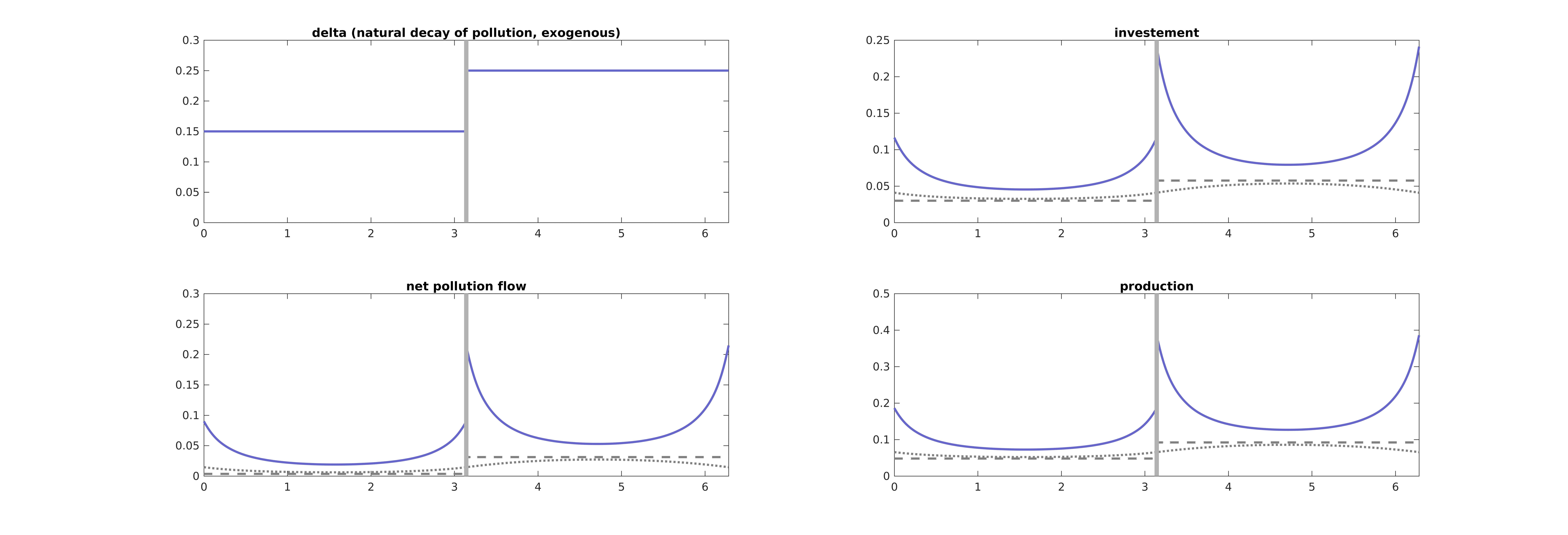}
	\caption{\footnotesize
		The case of two players controlling one half of the circle with different levels natural decay of the pollution $\delta$. The value of $\delta$ in the territory of the first player (on the left) is $\delta_1$ (constant it her territory) equal to $0.15$ while the value of $\delta_2$ is $0.25$. All other parameters (are constant over the space and) are the same as in Figure \ref{fig:2pl-tuttoconst-sect4}. Continuous (and colored) lines represent the Nash equilibrium, the dotted (and gray) lines are the cooperative benchmark (with the same $\sigma$) while the dashed (gray) lines are the zero-diffusion benchmark.}
	\label{fig:heterogeneousdelta}
\end{MyFigure}

\subsection{Geographic ecological discrepancy}
In Figure \ref{fig:heterogeneousdelta} we look at the situation where the territories of the two regions have different natural decays of  pollution. The effect of increasing $\delta$ on investment is positive (see Proposition \ref{prop:rhodelta}). The intuition is rather straightforward: a higher value of $\delta$ reduces, for a given investment strategy, the future stock of pollution and then it reduces the marginal disutility of polluting with respect to marginal utility of consumption. So it tends to increase input use and thus  production and consumption. Conversely, as it transpires from the expression of the optimal level of depollution expenditures (\ref{bopt-sect2}), the latter are not impacted by variations on $\delta$ so differences in input use drive mechanically the differences in net emissions flow and production.

We can also note that $\delta$ enters in the expression of equilibrium investment (\ref{iopt-sect2}) only through the values of $\alpha_{j}$. In this equation $\delta$ sounds as a substitute to $\rho$. This is not so surprising because both parameters have intertemporal implications: they act to discount the future effects of  present actions . Nevertheless here  substitutability is particularly large (1 to 1 at any time) because, differently from standard growth model, the decay $\delta$ directly acts on a variable (pollution) that linearly appears in the utility function. So all the previous remarks on the effects of $\delta$ on various endogenous variables, can be replicated exactly for $\rho$.

Differently from what we had for instance in the examples of Section \ref{se:border} we can observe that in Figure \ref{fig:heterogeneousdelta} the two benchmarks we use (the cooperative case with the same $\sigma$ and the no-diffusion situation) give distinct profiles. The same behavior will appear in Figure \ref{fig:heterogeneousw}. Both benchmarks are associated with less production and emissions than the Nash equilibrium (this is not surprising because in both cases the negative effects of pollution are completely internalized) but the optimal choice of the planner is to pollute more compared with the zero-diffusion benchmark in the low-delta zone. The reason behind is rather simple: the planner knows that part of the emissions will move to high-delta part and will therefore decay more quickly. This mechanism also explains why the difference between the two is higher as we get closer to the border. A symmetric argument is enough to figure out why the investment-choice of the planner tends to be lower relative to  the zero-diffusion benchmark in the high-delta zone.


\subsection{Geographic discrepancy in preferences}
In Figure \ref{fig:heterogeneousw} we represent again a symmetric two-country example with all the parameters space-independent except the unitary disutility of pollution $w$ which is $0.9$ in the region controlled by the first player and $1.1$ for the other. As one can infer from Proposition \ref{prop:rhodelta},  pollution decreases when $w$ is bigger: a larger marginal disutility from pollution leads the local planner to reduce investment and consumption to be able to reduce emissions. This fact also involves a reduction in production and in the net flow of emissions.

\begin{MyFigure}[H]
	\centering
	\includegraphics[width=\linewidth]{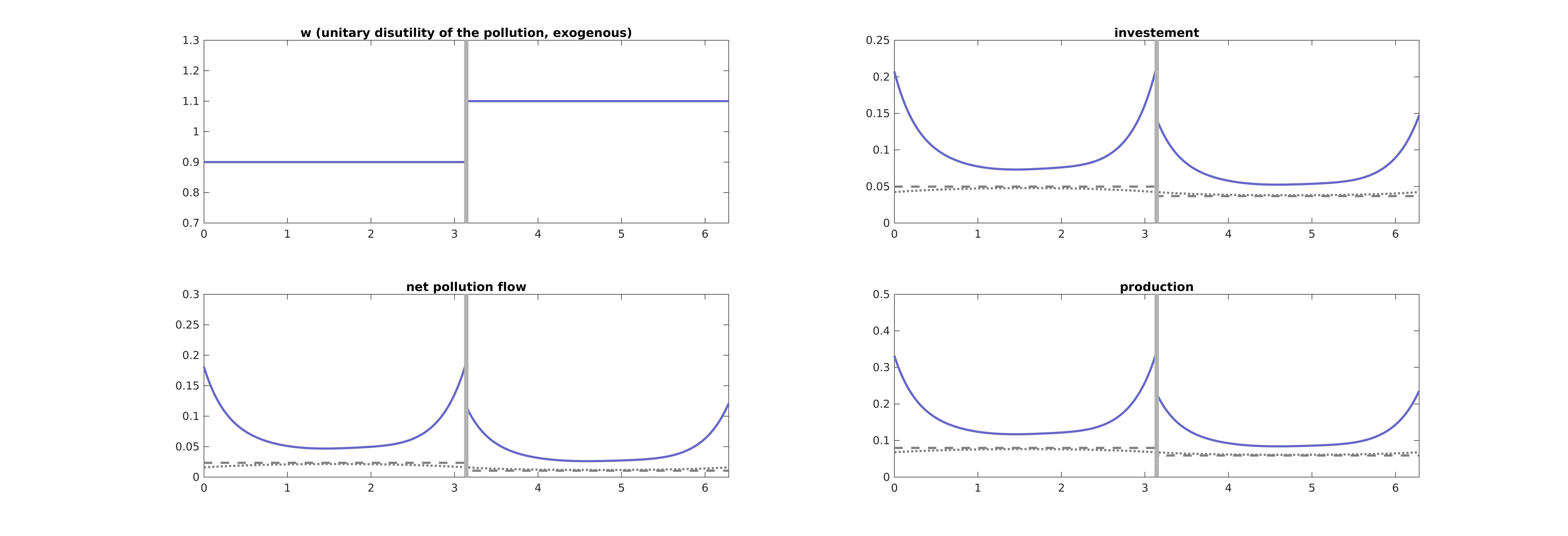}
	\caption{\footnotesize
		The case of two players controlling one half of the circle with different unitary disutility of the pollution $w$. The value of $w$ in the region controlled by the first player (on the left) is $w_1$ (constant it her territory) equal to $0.9$ while the value of $w_2$ is $1.1$. All other parameters (are constant over the space and) are the same as in Figure \ref{fig:2pl-tuttoconst-sect4}.
Continuous (and colored) lines represent the Nash equilibrium, the dotted (and gray) lines are the cooperative benchmark (with the same $\sigma$) while the dashed (gray) lines are the zero-diffusion benchmark.}
	\label{fig:heterogeneousw}
\end{MyFigure}


\section{Advection}
\label{se:advection}

So far we have considered, for the natural spatiotemporal dynamics of pollution, a completely homogeneous diffusion/spreading process. Indeed, if we abstract away from the agents' decisions, the dynamics described by (\ref{eq:stateequation-sect2}) reduces to
\[
\displaystyle{\frac{\partial p}{\partial t}(t,x) = \sigma \frac{\partial^2 p}{\partial x^2} (t,x)  - \delta(x) p(t, x)}.
\]
In the right side of this expression only the second derivative term describes the spatial dynamics of pollutants while the decay term is, essentially, purely local.

A more general and more realistic formulation is possible by adding to the right hand side above a term representing exogenous location-specific flow, i.e.
$${v}(x) \frac{\partial p}{\partial x}(t,x).$$
This new term is a vector field on $S^1$ specifying at each location a further movement term (which adds to the already described diffusion term) having speed $v(x)$ at any point $x$. It is called \emph{advection} in atmospheric physics\footnote{{A more general formulation would be to consider $v$ depending also on time: $v(t,x)$. Of course this kind of generality can be interesting for a variety of situations but unfortunately our approach cannot include a time dependent $v$ since we need the operator $\L$ defined in Appendix \ref{sub:operatos} to be time-independent.}}. Advection allows to take into account the fact that basic dispersion of pollutants need not to be ``side-homogeneous'', for instance due to winds, currents or geographic characteristics. {It has been jointly studied and modelled with diffusion in atmospheric physics since the beginning of the 20th century, see for example Roberts (1923). It is duely accounted for in the recent transboundary pollution literature either in static frameworks (like in Lipscomb and Mubarak, 2017) or dynamic ones (see de Frutos et al., 2021).}
With this new term the evolution equation \eqref{eq:stateequation-sect2} for the pollution stock becomes
\begin{small}\begin{equation}\label{SE-cest5}
\begin{cases}
\displaystyle{\frac{\partial p}{\partial t}(t,x) =\sigma \frac{\partial^2 p}{\partial x^2}
+ {v}(x) \frac{\partial p}{\partial x}(t,x)  - \delta(x) p(t, x) +  i(t,
x)-\eta b(t,x)^{\theta},
 \ \ (t,x) \in\R_+\times S^1},\\\\
p(0,x)=p_0(x),  \ \ \ x\in S^1.
\end{cases}
\end{equation}
\end{small}

All the results provided in the previous sections can be generalized to the system including a generic advection term. In particular the counterpart of the Nash equilibrium described in Theorem \ref{th:oo-sect2} reads now as
\begin{equation}\label{bopt-sect5}
b^{ad,*}_j(t,x)= \left[(A_j(x)-1)\eta\theta\right]^{\frac{1}{1-\theta}},
\end{equation}
\begin{align}\label{iopt-sect5}
i^{ad,*}_j(t,x)
= \alpha_j(x)^{-\frac{1}{\gamma_j}}(A_j(x)-1)
^{\frac{1-\gamma_j}{\gamma_j}}
+\left(\eta\theta\right)^{\frac{1}
{1-\theta}}(A_j(x)-1)^{\frac{\theta}
{1-\theta}},
\end{align}
where $\alpha_j$ is now  the unique solution to the ODE
\begin{equation}\label{ODEalpha2-sect5}
\rho_j\alpha_j(x)- \sigma \alpha_j''(x)+v(x)  \alpha_j'(x)+(v'(x)+\delta(x))\alpha_j(x)=\widehat{w_j}(x), \ \ \ x\in S^1.
\end{equation}

\begin{MyFigure}[H]
\centering
\includegraphics[width=\linewidth]
{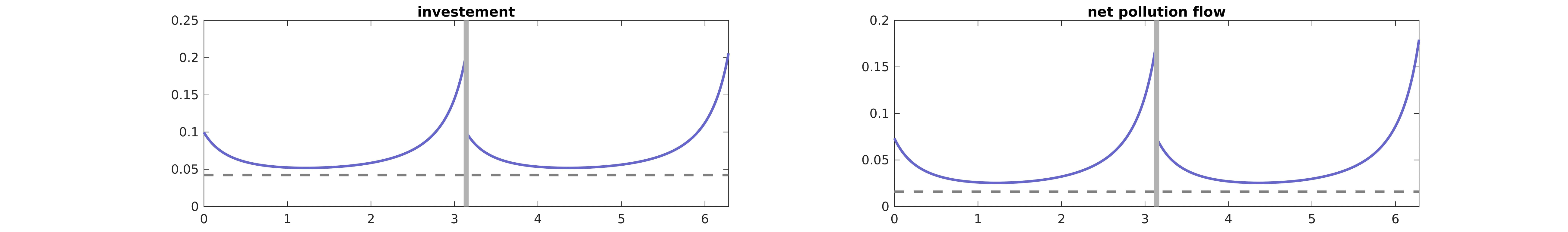}
\caption{\footnotesize
The case of two players controlling one half of the circle. All the setting and the parameters values are the same as in Figure \ref{fig:2pl-tuttoconst-sect4} but now a spatially constant advection term $v(x)=-0.08$ is introduced.
The continuous (and blue) lines represent the Nash equilibrium while the dashed (and gray) lines are the cooperative benchmark or, equivalently, the non-diffusion benchmark.}
\label{fig:dilution}
\end{MyFigure}

To understand how the results can be qualitatively affected by the presence of the advection term we see, in Figure \ref{fig:dilution}  how the spatial profiles reported in Figure \ref{fig:2pl-tuttoconst-sect4} are altered when a constant advection term with $v(x) = v<0$ is introduced (all other specifications and parameterization remaining the same). This, concretely, means introducing a persistent current which is clockwise when we look at the $S^1$ as a subset of $\mathbb{R}^2$, that is, roughly speaking a current directed to the ``right'' in the figure. As a consequence, pollutants tend to move quicker to the other region if they are produced in a location which is close to the right borders (say \emph{east}) and vice versa. For this reason the authority internalizes even less the disutility due to emissions coming from eastern locations and so she has un incentive to produce and pollute more there. A symmetric argument explains why the production and the emission in the western locations are lower than in the no-advection case of Figure \ref{fig:2pl-tuttoconst-sect4}.


\section{Conclusion}

This paper is the first outcome of a larger research program on the theoretical characterization of a wide class of problems with free riding under natural diffusion and spatial externalities. This class of problems is arguably of particular importance given the pressing ongoing environmental and epidemic crises under way. We have exposed the main principles of a modelling them as infinite-dimensional differential games in continuous time and space.
We have fully explored the case of the free riding pollution problem, a preeminent representative of the latter class. Indeed we have  revisited the foundations of the spatial externalities theory in this case, a central research line not only in the most recent literature on environmental federalism (see for example Hutchinson and Kennedy, 2008) but also in several recent much more applied works on transboundary pollution (see  Lipscomb and Mobarak, 2017, as a representative example). In particular, precisely to close the gap between the theory and the latter empirical works, we consider a spatiotemporal framework where, instead of assuming \emph{ad hoc} pollution diffusion schemes across space, we use a realistic spatiotemporal law of motion for air and water pollution (natural diffusion and advection). This has led us to what we believe to be a strong methological contribution since we have ultimately to tackle and to solve spatiotemporal non-cooperative (and cooperative) differential games, which is far more complicated than the counterpart static games in the benchmark theory. {As repeatedly outlined above, the mathematical literature is very scarce in this area, and it's even much scarcer in economics where the work of de Frutos and co-authors discussed in the introduction is a notable exception.} We also incorporate into the analysis a large set of discrepancies across state and jurisdictions, which broadens even more the scope of our theory and its practical interest. At the end of the day, we have been able not only to find out a closed-form characterization of the spatiotemporal equilibria but also to identify a rich set of spatial patterns taking advantage of the many heterogeneities accommodated by our methodology (the most elementary being the asymmetry of players). This has in turn allowed us to check that our model is consistent with the set of stylized facts put forward by the related empirical literature.

Several additional questions may be tackled. For example, our setting is general enough to accommodate the two global governance levels (federal and supranational). Furthermore, since our analytical method allows for deep geographic discrepancies, this in principle enables us to address some of the hottest questions in the international agenda, in particular those related to the North/South environmental divide. Beyond the fundamental pollution free riding problem studied in this paper, we are in the course of applying our methodology to other problems in the more general class described above, in particular to epidemic spatiotemporal dynamics. Of course, this cannot be in principle done without adaptation, and full closed-form characterization may not be possible, making the use of complementary numerical methods indispensable. We are in this process.


\bigskip



\bigskip

\bigskip

\bigskip

\begin{footnotesize}

\bigskip

\bigskip

\bigskip

\bigskip

\bigskip


\appendix

\section*{Appendices}
In the following appendices we describe and we study, in a rigorous mathematical way, the generalized problem with advection presented in Section \ref{se:advection}, which includes also the other cases presented in the paper.

\section{Formulation of the problem and standing assumptions}
\label{appsec:formulation}

Let $S^1$ be the unitary circle in $\R^2$:
$$S^1:=\big\{x\in\R^{2}: \ |x|_{\R^{2}}=1\big\}.$$
Hereafter, we often identify $S^1\cong 2\pi\R/\Z$. Accordingly,  we identify functions $\varphi:S^1\to\R$ with $2\pi$-periodic function $\varphi: \R\to \R$. Moreover, given a function $\psi:\R_+\times S^1\to\R$, we denote by $\psi_t$ and $\psi_x$, respectively, the derivative with respect to the first (time) and the second (space) variable.

The following useful objects with the associated assumptions are given\footnote{For the definition of  the Lebesgue spaces $L^q$ we refer the reader e.g. to Brezis (2011), Chapter 4. We recall that these spaces are done by equivalence classes of functions according to the equivalence relation which identifies functions which are equal  \emph{almost everywhere}.}.
\begin{Assumption}\label{ass:data}
	\begin{enumerate}
		\item[]
		\item
		$p_0\in L^2(S^1;\R_+)$;
		\item $\delta\in C(S^1;\R_+)$;
		\item $v\in C^1(S^1;\R)$;
		\item $A_j\in L^\infty(M_j;\R_+)$ and there exists a constant $l$  such that
		$1<l\leq A_j(x)$ for all $j=1,\dots,N$;
		\item for each $j=1,...,N$, one has  $w_j\in C(M_j;\R_+)$ and $w_j$ be extended to a function $\overline{w_j}\in C(\overline{M_j};\R)$ such that $\overline{w_j}(x)>0$ for each $x\in \overline{M_j}$.
	\end{enumerate}
\end{Assumption}

We consider the following {multiagent problem} in $S^1$. We fix a positive integer number $N\geq 1$ representing the number of players. Each player $j=1,...,N$ is endowed with her/his own part of territory $M_j\subset S^1$, assuming that 
$$
M_j\cap M_h=\emptyset \ \ \ \ \mbox{for} \  h\neq j, \ \ h,j=1,...,N.
$$
{Notice that there may be parts of $S^1$ that do not belong to any player.}
{We assume that $M_j$ is connected and relatively open {(hence an open interval in the circle)} for every $j =1,\dots, N$.}
{We set
	$$M_0:=S^1\setminus\bigcup_{j=1}^N M_j.$$ Note that $M_0$ is closed and contains at least the boundary points of $M_j$, for $j=1,\dots,N$. From now on, when we take an index $j$ without mentioning explicitly where it lies, we mean that $j\in\{1,\dots,N\}$.}

Player $j$ chooses the investment policy $i_j(t,x)$ and the depollution policy $b_j(t,x)$ at time $t\in\R_+$ and location $x\in M_j$.
These functions belong, for $j=1,\dots,N$, to the class of functions 
\begin{footnotesize}
\begin{equation}
\label{eq:Cj}
\begin{split}
\mathcal{A}_j:=&\Bigg\{(i_j,b_j): 	\R_+\times M_j\to \R_+\times \R_+ \ s.t.\\& \ \ \ \ \ \ \ \ \ \  \ t\mapsto (i_j(t,\cdot),b_j(t,\cdot)) \in L^1((\R_+,e^{-\rho_j t} dt) ;L^2(M_j,\R_+))\times L^1((\R_+, e^{-\rho_j t} dt);L^2(M_j,\R_+))\\  & \ \ \ \ \ \ \ \ \ \ \ \ \ \ \ \ \ \ \ \ \ \ \  \mbox{and} \ \  (A_j(x)-1) i_j(t, x)-b_j(t,x)\geq 0 \ \ \ \ \mbox{for a.e.}  \  (t,x)\in \R_+\times M_j\Bigg\}.
\end{split}
\end{equation}
\end{footnotesize}
Given, for every $j=1,\dots,N$, an element $(i_j,b_j) \in \mathcal{A}_j$, we obtain
a couple $(i,b):\R_+\times S^1 \to\R^2$ defined as in \eqref{def:ib} and
$$(i,b)\in \mathcal{A}:= \mathcal{A}_1\times ... \times \mathcal{A}_N.$$
{Let $\sigma>0$,  $\theta\in(0,1)$, $\eta\geq 0$ be given constants}.
The evolution of the state variable $p(t,x)$, representing pollution, is formally given by
the following parabolic PDE:
\begin{equation}\label{SE}
\begin{cases}
\displaystyle{p_t(t,x) = \sigma p_{xx}(t,x)+  {v}(x) p_x(t,x)  - \delta(x) p(t, x) +  i(t,
	x)-\eta b(t,x)^{\theta},
	\ \ (t,x) \in\R_+\times S^1},\\\\
p(0,x)=p_0(x),  \ \ \ x\in S^1,
\end{cases}
\end{equation}
where
{\begin{equation}\label{def:ib}
	\big(i(t,x),b(t,x)\big):=\begin{cases}
	\left(i_j(t,
	x), b_j(t,x)\right), \ \ \ \mbox{if} \ x\in M_j,\, j=1,\dots,N\\
	0, \ \ \ \ \ \ \ \ \ \ \ \ \ \ \, \ \ \ \ \ \ \ \ \ \mbox{if} \ x\in M_0=S^1\setminus\bigcup_{j=1}^N M_j.
	\end{cases}\end{equation}}
	{Existence and uniqueness of solutions to this PDE when the couple of functions $(i,b)$ is defined as above can be proved with classical arguments for parabolic equations, see e.g. Lions and Magenes (1972), Theorem 6.1, p.\,33.} {However, we will look at it as an evolution equation in a suitable Hilbert space (see Appendix \ref{appsec:reformulation}) and will define the concept of mild solution, which will be our  solution concept.}

%
%
%
%
%
Next, given $j=1,...,N$, we denote
$$
(i_{-j},b_{-j})= \left((i_1,b_1), ..., (i_{j-1},b_{j-1}), (i_{j+1},b_{j+1}),...,(i_N,b_N)\right).
$$
The payoff functional of player $j$ associated to the choice $(i_{j},b_{j)})$ when the other players undertake the strategies $(i_{-j},b_{j})$ is
\begin{equation}\label{F}
J_j^{(i_{-j},b_{-j})}\big(p_0;(i_j,b_j)):=\int_0^{\infty} e^{-\rho_j t} \left(\int_{M_j}\left(
\frac{\big((A_j(x)-1) i_j(t, x)-b_j(t,x)\big)^{1- \gamma_j}}{1-\gamma_j}
- w_j(x)p(t, x) \right) dx\right)dt,
\end{equation}
where $\rho_j>0$, $\gamma_j\in(0,1)\cup(1,\infty)$.

%


\begin{Remark}
	In our model the space variable lies in the one dimensional circle $S^1$. 
	Extensions to different space structures are possible but the one dimensional case allows to compute more easily the solution and to perform a more precise analysis of its behavior.
\end{Remark}

\section{Reformulation of the problem in infinite dimension}
\label{appsec:reformulation}

We are going to rigorously rewrite the above problem in a suitable infinite dimensional space. {The reformulation in itself has not a specific economic interpretation but it is merely a technical tool in order to be able to apply to our problem the theory of evolution differential equation in infinite dimensional spaces (see, as a general reference, Bensoussan et al., 2006).}
First, we  provide the various ingredients of this reformulation: the spaces, the operators,
and the reformulation itself.

\subsection{The spaces}

Consider the spaces
$$H:=L^2(S^1):=\left\{f:S^1\to\R \ \mbox{measurable}: \ \int_ {S^1} |f(x)|^2dx<\infty\right\}$$
{and, for $j=0,\dots,N$,}
$$H_j:=L^2(M_j):=\left\{f:M_j\to\R \ \mbox{measurable}: \ \int_ {M_j} |f(x)|^2dx<\infty\right\},$$ endowed with the with inner products
$$
\langle f,g\rangle_H :=\int_{S^1} f(x)g(x) dx, \ \ \ f,g\in H,$$
$$
\langle f,g\rangle_{H_j} :=\int_{M_j} f(x)g(x) dx, \ \ \ f,g\in H,$$
which render them separable Hilbert spaces.
Denote by $|\cdot|_H$, $|\cdot|_{H_j}$ the associated norm, i.e.
$$|f|_H^2:=\int_{S^1} |f(x)|^2dx, \ \ \  \ f\in H,$$
$$
|f|_{H_j}^2:=\int_{M_j} |f(x)|^2dx, \ \ \  \ f\in H_j.$$
Finally denote by $H^+,H_j^+$ the positive cones of $H$ and $H_j$, respectively.
{
	Now, for every $f \in H$ we can write
	\begin{equation}\label{eq:directdsumibnew}
	f(x):=\sum_{{j=0}}^N f(x){\bf 1}_{M_j}(x), \quad x \in S^1.
	\end{equation}
	Since the restriction of $f{\bf 1}_{M_j}$ to $H_j$ is an element of $H_j$
	we  write, again with a slight abuse of notation
	$$
	H=\bigoplus_{{j=0}}^N H_j, \ \ \ \ H^+=\bigoplus_{{j=0}}^N H_j^+.
	$$}

\subsection{The operators}
\label{sub:operatos}

Denote by
$L(H)$ the space of bounded linear operators on $H$. Consider the differential operator
$(\L,D(\L))$ in $H$, where
$$D(\mathcal{L})=W^{2,2}(S^1;\R);$$
$$[\mathcal{L} \varphi](x)=\sigma \varphi''(x)+v(x)\varphi'(x)-\delta(x)\varphi(x),
\ \ \  \varphi\in D(\L).$$
The latter is a closed, densely defined, unbounded linear operator on the space $H$ (see, e.g. Lunardi, 1995, page 72, Section 3.1.1). A core\footnote{{A core is a subspace the domain $D(\mathcal{L})$ which is is dense in $D(\mathcal{L})$ for the graph norm $|x|_{\mathcal{L}} := |x| + |\mathcal{L}x|$.}} for it is the space $C^\infty(S^1;\R)$
(see, e.g., Engel and Nagel, 1995, pages\,69-70).
Integration by parts shows that 
\begin{equation}\label{qq}\langle \mathcal{L}\varphi,\psi\rangle_H =\langle \varphi,\mathcal{L^\star}\psi\rangle_H, \ \ \ \forall \varphi,\psi\in C^\infty(S^1;\R)
\end{equation}
where
$D(\mathcal{L^\star})=D(\mathcal{L})=W^{2,2}(S^1;\R)$
and
\begin{equation}\label{eq:L*}
[\mathcal{L}^\star\psi](x)= \sigma  \psi''(x)-v(x)\psi'(x)-(v'(x)+\delta(x))\psi(x),
\ \ \  \psi\in C^\infty(S^1;\R).
\end{equation}
Since $C^\infty(S^1;\R)$ is a core for $\mathcal{L}$ and $\mathcal{L}^\star$, \eqref{eq:L*} extends to all couples of functions in $D(\mathcal{L})$. This shows that $\mathcal{L}$ is self-adjoint and dissipative if $v\equiv 0$.

Integration by parts also shows
\begin{equation}\label{eq:intpartv}
\int_{S^1} v(x) \varphi(x)\varphi'(x) \d x=-\int_{S^1} v(x) \varphi(x)\varphi'(x) \d x- \int_{S^1} v'(x) |\varphi(x)|^2 \d x,
\end{equation}
hence,
\begin{equation}\label{eq:intpartvbis}
\int_{S^1} v(x) \varphi(x)\varphi'(x) \d x=- \frac{1}{2} \int_{S^1} v'(x) |\varphi(x)|^2 \d x
\end{equation}
{Next, we compute, again using integration by parts (in particular, \eqref{eq:intpartv} and \eqref{eq:intpartvbis}),
	\begin{eqnarray*}
		\langle \mathcal{L}\varphi,\varphi\rangle_H &=&\int_{S^1} \left([\mathcal{L}\varphi](x)\right)\varphi(x)\d x\\
		&=&
		-\int_{S^1} \sigma(x)|\varphi'(x)|^2 \d x -\int_{S^1} v(x) \varphi(x)\varphi'(x) \d x-\int_{S^1}(v'(x)+\delta(x))|\varphi(x)|^2\d x\\
		&=&
		-\int_{S^1} \sigma(x)|\varphi'(x)|^2 \d x -\int_{S^1}\left(\frac12v'(x)+\delta(x)\right)|\varphi(x)|^2\d x\\
		&\leq &\left|\left(\frac{1}{2}v'+\delta\right)\wedge 0\right|_\infty \,|\varphi|_H^2.
	\end{eqnarray*}
	Hence, the operator $\mathcal{L}$ is pseudo-dissipative, and so is $\mathcal{L}^\star$.
	Therefore,
	by Engel and Nagel (1995) (see, in particular, Chapter\,II), we
	see that $\L$ generates a strongly continuous  semigroup $(e^{t\mathcal{L}})_{t\geq 0}\subset L(H)$.}
\subsection{State equation}
\label{subapp:refstateequation}

Setting
$$I_j(t):=i_j(t,\cdot), \ \ \  \ B_j(t):=b_j(t,\cdot), \ \ \ j=1,...,N,$$
we see that
$(I_j,B_j):\R_+ \to H_j^+\times H_j^+$ and rewrite the set $\mathcal{A}_j$, given in \eqref{eq:Cj} (and then, consequently, $\mathcal{A}$), as follows:
\begin{equation*}
\begin{split}
\mathcal{A}_j=\,  &
\Bigg\{(I_j,B_j):\R_+ \to H_j^+\times H_j^+ \ \mbox{measurable}: \
\int_0^\infty e^{-\rho_j t}  (|I_j(t)|^2_{H_j}+|B_j(t)|^2_{H_j})dt<\infty,
\\ & \ \ \ \ \ \ \ \ \ \ \ \ \ \  \ \ \ \
(A_j(\cdot)-1) I_j(t)(\cdot)-B_j(t)(\cdot)\geq 0 \ \ \ \
\mbox{for a.e.} \ t\in\R_+\Bigg\}.
\end{split}
\end{equation*}
The  last inequality above means that, for a.e. $t\ge 0$, the function in $H_j^+$ given by
$$x \mapsto (A_j(x)-1) i_j(t,x)-b_j(t,x)$$ is nonnegative for a.e. $x\in M_j$. 

Given  $\big((I_j,B_j)\big)_{j=1,...,N}\in\mathcal{A},$ we set
$$
I(t):=0\cdot \mathbf{1}_{M_0}(x)+\sum_{j=1}^N\mathbf{1}_{M_j} I_{j}(t), \ \ \ B(t):=0\cdot \mathbf{1}_{M_0}(x)+\sum_{j=1}^N \mathbf{1}_{M_j}B_j(t), $$
and
$$
[B(t)^{\theta}](x):=(B(t)(x))^{\theta}, \ \ \ x\in S^1.
$$
{Then, with the formal identification}
$$[P(t)](x):=p(t,x), \ \ \ x\in S^1,$$
we reformulate \eqref{SE} in $H$ as
\begin{equation}\label{SEi}
\begin{cases}
P'(t)= \L P(t) + I(t)-\eta B(t)^{\theta},
\ \ \ \ t\geq 0,\\
P(0)=p_0\in H.
\end{cases}
\end{equation}
According to Bensoussan et al. (2007) (Part\,II,\,Chapter\,1.\,Definition\,3.1(v)), we define the \emph{mild solution} to \eqref{SEi} as
\begin{equation}\label{eq:Ymild}
P(t)=e^{t\L} p_0+\int_0^t e^{(t-s)\L}
\left[I(s)-\eta B(s)^{\theta}\right]ds, \ \ \ t\geq 0.
\end{equation}
This will be our concept of solution to the state equation.

{{We present the following lemma, which is needed to prove uniqueness of affine Markovian equilibria. 
\begin{Lemma}\label{lemma:phi}
Let $\Phi\in L(H;H_{j})$ be positivity preserving and let $h\in H^{+}\setminus\{0\}$. 
If the function $$ \R_{+}\to H_{j}, \ \ t\mapsto \Phi \int_0^t e^{(t-s)\L}h\, ds$$  is constant, then $\Phi=0$.
\end{Lemma}
\begin{proof}
First, notice that $ \int_0^t e^{(t-s)\L}h\,ds= \int_0^t e^{s\L}h\,ds$. 
Let $\Phi^{\star}:H_{j}\to H$ be the adjoint of $\Phi$; clearly it is positive preserving too. Given $v\in H_{j}$, the function 
$$t\mapsto\Big\langle\Phi \int_0^t e^{s\L}h \,ds, v\Big\rangle_{H_{j}}= \Big\langle \int_0^t e^{s\L}h\, ds, \Phi^{\star}v\Big\rangle_{H}$$ is constant. Therefore, in particular, \begin{align*}
0&=\frac{d}{d t} \ \Big\langle \int_0^t e^{s\L}h\, ds, \Phi^{\star}v\Big\rangle_{H}\Big|_{t=1}= \Big\langle \frac{d}{dt} \int_0^t e^{s\L}h \,ds\Big|_{t=1}, \Phi^{\star}v\Big\rangle_{H} 
= \Big\langle e^\L h,  \Phi^{\star}v\Big\rangle_{H} 
\end{align*}
The semigroup $\{e^{t\L}\}_{t\geq 0}$ maps nonnegative not identically null functions into strictly positive ones, hence $e^{\L}h$ is strictly positive.
 Hence, since  $\Phi^{\star}$ is positivity preserving, from the latter  equality we see that 
$\Phi^{\star}v=0$ for every $v\in H^{+}_{j}$. Using the decomposition $v=v^{+}- v^{-}$,  we also get  $\Phi^{\star}v=0$ for every $v\in H_{j}$, i.e. that $\Phi^{\star}=0$, from which the claim follows.  
\end{proof}}}

\subsection{Objective functionals}
Now we rewrite the objective functionals \eqref{F} of the players.
Define
$$
\left[\frac{\big((A_j-1) I_j(t)-B_j(t)\big)^{1- \gamma_j}}{1-\gamma_j}\right](x):=\frac{\big((A_j(x)-1) i_j(t, x)-b_j(t,x)\big)^{1- \gamma_j}}{1-\gamma_j}, \ \ \forall (t,x)\in \R_+\times M_j.
$$
Then, 
\begin{equation}
\label{eq:utnew}
\int_0^{\infty} e^{-\rho_j t} \left(\int_{M_j}
\frac{\big((A_j(x)-1) i_j(t,x)-b_j(t,x)\big)^{1-\gamma_j}}{1-\gamma_j}
dx\right) dt=
\int_0^{\infty} e^{-\rho_j t}
\left\langle\frac{\big((A_j-1) I_j(t)-B_j(t)\big)^{1- \gamma_j}}{1-\gamma_j},\mathbf{1}_{M_j}\right\rangle_{H_j}dt.
\end{equation}
On the other hand, 
$$
\int_0^{\infty} e^{-\rho_j t} \left(\int_{M_j}w_j(x)p(t, x) dx\right)dt=  \int_0^{\infty} e^{-\rho_j t}
\big\langle\widehat{w_j},P(t)\big\rangle_H dt,
$$
where
$\widehat{w_j}:S^1\to \R$ is defined as
\begin{equation}
\label{eq:defwhatj}
\widehat{w_j}(x):=\begin{cases} w_j(x), \ \ \ \mbox{if} \ x\in M_j,\\
0, \ \ \ \ \ \ \ \ \ \mbox{if} \ x\notin M_j.
\end{cases}
\end{equation}
Set
$$
e^{-(\rho_j-\L)t}:=e^{-\rho_j t} e^{t\L}, \ \ \ t\geq 0.
$$
By \eqref{eq:Ymild}
we have
\begin{equation}
\label{Finew}
\begin{split}
&
= \int_0^{\infty} e^{-\rho_j t}
\big\langle\widehat{w_j},P(t)\big\rangle_H dt
\\&= \int_0^{\infty} e^{-\rho_j t}
\left\langle \widehat{w_j},e^{t\L} p_0+\int_0^t e^{(t-s)\L}
\left[I(s)-\eta B(s)^{\theta}\right]ds\right\rangle_H dt
\\&
= \left\langle\widehat{w_j}, \int_0^{\infty} e^{-(\rho_j -\L) t}p_0\, dt\right\rangle_H+
\int_0^{\infty} e^{-\rho_j  t} \left\langle \widehat{w_j},\int_0^t e^{(t-s)\L}
\left[I(s)-\eta B(s)^{\theta}\right]ds\right\rangle_H dt
\end{split}
\end{equation}
Hence, the original functional $J_j$ of agent $j$ rewrites as  
 \begin{eqnarray}\label{eq:Finfdim}
J_j^{(I_{-j},B_{-j})}\big(p_0;(I_j,B_j))&=&
\left\langle\widehat{w_j}, \int_0^{\infty} e^{-(\rho_j -\L) t}p_0\, dt\right\rangle_H + \int_0^{\infty} e^{-\rho_j t}
\left\langle\frac{\big((A_j-1) I_j(t)-B_j(t)\big)^{1- \gamma_j}}{1-\gamma_j},\mathbf{1}_{M_j}\right\rangle_{H_j}dt\nonumber 
\\[3mm]
&&+
\int_0^{\infty} e^{-\rho_j  t} \left\langle \widehat{w_j},
\int_0^t e^{(t-s)\L}\left[I(s)-\eta B(s)^{\theta}\right]ds\right\rangle_H.
\end{eqnarray}
The reformulated problem for the agent $j$
consists then in maximizing the functional $J_j$ in \eqref{eq:Finfdim}, over the set $\mathcal{A}_j$ and
under the state equation \eqref{SEi}.
Note  that, in this reformulation, the first term of the functional $J_j$
is the only one which depends on the initial datum.

{\section{Reformulation of the functionals and analysis of the solution}}
\label{appsec:Nash}

\subsection{The functions $\alpha_j$ and their properties}\label{app:alpha} 
The following assumption, which will be a standing assumption from now on, ensures that $\rho_j$ belongs to the resolvent set of the operator $\L$ (see Engel and Nagel, 1995).
\begin{Assumption} \label{hp:rhoj}
	For all $j=1,\dots,N$ we have
	$$
	\rho_j > \left|\left(\frac{1}{2}v'+\delta\right)\wedge 0\right|_\infty.
	$$
\end{Assumption}
By Assumption \ref{hp:rhoj}, the operator $$\rho_j-\mathcal{L}:\ D(\mathcal{L})\,\longrightarrow\, H$$ is invertible with bounded inverse $(\rho_j-\L)^{-1}:H\to H$ and
\begin{equation}\label{risol}
\displaystyle{{(\rho_j -\L)^{-1}}h=\int_0^{\infty} e^{-(\rho_j-\L)t}h\,dt}\ \ \  \ \ \ \forall h\in H.
\end{equation}
We define
\begin{equation}\label{defalpha}
\alpha_j:=(\rho_j-\L^{*})^{-1}\, \widehat{w_j}\in D(\L)= W^{2,2}(S^1;\R).
\end{equation}
By definition $\alpha_j$ is therefore the unique  solution in ${D(\L^{*})}=W^{2,2}(S^1;\R)$ of the abstract ODE
\begin{equation}\label{eq:alpha}
{\left(\rho_j-\mathcal{L}^{*}\right)}\alpha_j=\widehat{w_j}.
\end{equation}
More explicitly, $\alpha_j$, as defined in \eqref{defalpha}, is the unique solution in the class $W^{2,2}(S^1;\R)$ to
\begin{equation}\label{ODEalpha}
{\displaystyle{\rho_j\alpha_j(x)- \sigma \alpha_j''(x)+v(x)  \alpha_j'(x)+(v'(x)+\delta(x))\alpha_j(x)=\widehat{w_j}(x),} \ \ \ x\in S^1,}
\end{equation}
meaning that it verifies \eqref{ODEalpha} pointwise almost everywhere in $S^1$ (this will be, from now on, the meaning of solution to such equation).\footnote{The latter ODE can be also viewed as on ODE on the interval $[0,2\pi]$ with periodic boundary conditions:
	$$\begin{cases}
	{\displaystyle{\rho_j\alpha_j(x)- \sigma \alpha_j''(x)+v(x)  \alpha_j'(x)+(v'(x)+\delta(x))\alpha_j(x)=\widehat{w_j}(x),} \ \ \ x\in(0,2\pi),}\\\\
	\alpha_j(0)=\alpha_j(2\pi), \ \ \alpha_j'(0)=\alpha_j'(2\pi),
	\end{cases}
	$$
	falling into the Sturm-Liouville theory with periodic boundary conditions (see Coddington and Levinson, 2013).}  By Sobolev embedding $W^{2,2}(S^1;\R)\subset C^{1}(S^1;\R)$, so $\alpha_j\in C^1(S^1;\R).$
\smallskip

We state an a priori estimate for the solution of the equation $(\lambda-{\L^{*}})g=f$, when ${\inf_{{x\in S^{1}}}(\lambda +\delta(x)+v'(x))>0}$ and $f$ has a suitable regularity.
\begin{Proposition}
	\label{prop:mp}
	Let $M\subset S^1$ be an open nonempty interval. Let $f\in  L^2(S^1;\R)$ be such that $f|_{\overline{M}}$ and $f|_{S^1\setminus\overline{M}}$ are continuous, $f>0$ on $\overline{M}$, and $f=0$ on $S^1\setminus\overline{M}$. Finally, let $\lambda>0$ be such that ${\inf_{{x\in S^{1}}}(\lambda +\delta(x)+v'(x))>0}$ and  $g\in W^{2,2}(S^1;\R)$   such that $(\lambda-{\L^{*}})g=f$.
	Then there exists $\kappa>0$ such that
	$$\kappa\leq g\leq \max_{\overline{M}}f.$$
\end{Proposition}
\begin{proof}
	First we notice that, since $g\in W^{2,2}(S^1;\R)\subset C^1(S^1;R)$ and $S^1$ is compact,
	the function $g$ admits maximum and minimum over  $S^1$.
	
	\emph{Estimate from below.}  We identify $M$ with an open interval $(a,b)$, so $\{a,b\}=\partial M$.
	The fact that $(\lambda -\L)g=f$ and the assumption $\sigma>0$ yield
	$$ g''(x)={\frac{1}{\sigma}\left[(\lambda  +\delta (x)+v'(x))g(x) +v(x)g'(x)-f(x)\right]}, \ \ \ \mbox{for a.e.} \ x\in S^1.$$
	Since $g\in C^1(S^1;\R)$, it follows that
	$g\in C^2(S^1\setminus \partial M;\R)$
	and
	\begin{equation}\label{pp2}
	g''(x)={\frac{1}{\sigma}\left[(\lambda  +\delta (x)+v'(x))g(x) +v(x)g'(x)-f(x)\right]}, \ \ \ \forall   x\in S^1\setminus\partial M.
	\end{equation}
	Then, From \eqref{pp2} and Assumption \ref{ass:data}, we see that
	there exist finite   $g''(a^+):=\lim_{x\to a^+}g''(x)$ and  $g''(b^-):=\lim_{x\to b^-}g''(x)$ and their value are
	\begin{equation}\label{ppnew}
	g''(a^+)={\frac{1}{\sigma}\left[(\lambda  +\delta (a)+v'(a))g(a) +v(a)g'(a)-f(a)\right]},
	\end{equation}
	$$
	g''(b^-)={\frac{1}{\sigma}\left[(\lambda  +\delta (b)+v'(b))g(b) +v(b)g'(b)-f(b)\right]}.
	$$
	
	\noindent Let $x_*\in S^1$ be a minimum point of $g$ over $S^1$ and set $\kappa:=g(x_*)$.
	Clearly, since $g\in C^1(S^1,\R)$ it must be $g'(x_*)=0$.
	We distinguish three cases.
	
	\emph{Case 1: $x_*\in {M}$.}
	We have $g'(x_*)=0$ and $g''(x_*)\geq 0$.
	Plugging this into \eqref{pp2} we get
	$$
	(\lambda  +{\delta (x_*)+v'(x_{*}))}\kappa =\sigma g''(x_*)+f(x_*)> 0,
	$$
	hence, we conclude $\kappa>0$.
	
	\emph{Case 2: $x_*\in \{a,b\}$.} Assume, without loss of generality, that $x_*=a$.
	One has  $g'(a)=0$ and $\alpha_j''(a^+)\geq 0$.
	Plugging this into \eqref{ppnew} we get
	\begin{equation}\label{pp}
	0\leq {g}''(a^+)=\frac{1}{\sigma} \left[(\lambda  +\delta (a){+v'(a)})\kappa -f(a)\right],
	\end{equation}
	Since $f(a)>0$, we get $\kappa>0$.
	
	\emph{Case 3: $x_*\in S^1\setminus \overline{M}$.}
	In this case, as $f(x_*)=0$, arguing as before  we get
	$$
	{(\lambda  +\delta (x_*)+v'(x_*))}\kappa =\sigma g''(x_*)\geq 0,
	$$
	hence $\kappa\geq 0$. If $\kappa>0$, we have concluded. If $\kappa=0$, then the fact that $g(x_*)=g'_j(x_*)=0$ and \eqref{pp2} yield $g\equiv 0$ on $S^1\setminus {\overline{M}}$. By continuity of $g, g'$, we have also ${g}(a)={g}'(a)=0$. Moreover, it must be ${g}\geq \kappa=0$ on $\overline{M}$, as $x_*$ is a minimum point over $S^1$. Hence, it must be ${g}''(a^+)\geq 0$. Hence we agein get \eqref{pp}, a contradiction if $k=0$ as $f(a)>0$ by assumption.
	
	\emph{Estimate from above.} This part follows by arguments similar to the ones used for the estimate from below.
\end{proof}
Proposition \ref{prop:mp} immediatly yields the following two corollaries.
\begin{Corollary}\label{cor:pp}
	Let $\lambda>0$. The operator $(\lambda-{\L^{*}})^{-1}:H\to D({\L^{*}})\subset H$ is positivity preserving, i.e.
	$$
	f\in H, \ f\geq 0 \ \  \mbox{a.e.}  \ \ \ \Longrightarrow \ \ {(\lambda-\L^{*})^{-1}}f\geq 0 \ \mbox{a.e.}.
	$$
\end{Corollary}
\begin{proof}
	The claim immediately follows by Proposition \ref{prop:mp} due to density of $C(S^1;\R)$ in $H$.
\end{proof}

\begin{Corollary}\label{prop:max}
	Let Assumption \ref{ass:data} hold. Let $\alpha_j\in W^{2,2}(S^1;\R)\subset C^{1}(S^1;\R)$ be the solution to \eqref{ODEalpha}.
	There exists  $\kappa_j>0$ such that, for every $j=1,...,N$, we have   $$\kappa_j \leq \alpha_j\leq \mathcal{K}_j,$$ where
	$
	\mathcal{K}_j:=\max_{\overline{M_j}} \frac{1}{\rho_j} \overline{w_j}.
	$
\end{Corollary}

\begin{proof}
	Due to Assumption \ref{ass:data}, this is a direct application of Proposition \ref{prop:mp}.
\end{proof}

In the next results, we investigate the dependence of $\alpha_j$ on the data. We start with a convergence result on the diffusion coefficient $\sigma$.
\begin{Proposition}
	\label{pr:limitsigmato0toinfty}
	Let Assumption \ref{ass:data} hold.
	Denote by $\alpha_{j,{\sigma}}$ the unique solution to \eqref{ODEalpha} when $v(\cdot)\equiv 0$. We have
	$$
	\lim_{{\sigma}\to 0^+} \alpha_{j,{\sigma}}(x)=\frac{w_j(x)}{\rho_j+\delta(x)}, \ \ \ \ \lim_{{\sigma}\to +\infty} \alpha_{j,{\sigma}}(x)=\frac{\int_{S^1}w_j(x)dx}{\int_{S^1}(\rho_j+\delta(x))dx}, \ \ \ \forall x\in S^1.
	$$
\end{Proposition}
\begin{proof}

	\emph{Case $\sigma\to 0^+$.}
	First, notice that under our assumptions, \eqref{ODEalpha} reads as
	\begin{equation}\label{ODEalpha4}
	\displaystyle{\rho_j\alpha_{j,{\sigma}}(x)- {\sigma}\alpha_{j,{\sigma}}''(x)+\delta(x)\alpha_{j,{\sigma}}(x)=\widehat{w_j}(x), \ \ \ x\in S^1,}
	\end{equation}
	By Proposition \ref{prop:max} we have
	$$
	(\alpha_{j})_*(x):=\liminf_{\overline{{\sigma}}\to 0^+}  \big\{\alpha_{j,{\sigma}}(z): \ {\sigma}\leq \overline{{\sigma}}, \ z\in S^1, \ |z-x|\leq 1/\overline{{\sigma}}\big\}\geq 0, $$
	$$(\alpha_{j})^*(x):=\limsup_{\overline{{\sigma}}\to 0^+}  \big\{\alpha_{j,{\sigma}}(z): \ {\sigma}\leq \overline{{\sigma}}, \ z\in S^1, \ |z-x|\leq 1/\overline{{\sigma}}\big\}\leq \mathcal{K}_{j}.
	$$
	Clearly $(\alpha_{j})_*\leq (\alpha_{j})^*$.
	By stability of viscosity solutions (see e.g. Crandal-Ishii-Lions, 1992), the latter functions are, respectively, (viscosity) super- and sub-solution to the limit equation
	$$
	\displaystyle{\rho_j\alpha_{j,0}(x)+\delta(x)\alpha_{j,0}(x)=\widehat{w_j}(x)}
	$$
	whose unique solution is
	$$
	\alpha_{j,0}(x)=\frac{\widehat{w_j}(x)}{\rho+\delta(x)}.
	$$
	By standard comparison of viscosity solutions one has $(\alpha_{j})_*\geq \alpha_{j,0}\geq (\alpha_{j})^*$. It follows that
	$$\exists \lim_{{\sigma}\to 0^+} \alpha_{j,{\sigma}}(x)= (\alpha_{j})_*(x)= (\alpha_{j})^*(x)= \alpha_{j,0}(x) \ \ \ \forall x\in S^1.$$
	
	\emph{Case $\sigma\to +\infty$.}
	First, we rewrite \eqref{ODEalpha4} as
	\begin{equation}\label{ODEalpha5}
	\displaystyle{\alpha_j''(x)=\frac{1}{{\sigma}}\left[\rho_j\alpha_j(x)+\delta(x)\alpha_j(x)-\widehat{w_j}(x)\right], \ \ \ x\in S^1,}
	\end{equation}
	Notice now that $\alpha''_{j,\sigma}$ is equi-bounded and equi-uniformly continuous withe respect to $\sigma\geq 1$. Hence, by Ascoli-Arzel\`a Theorem we have that, from each sequence  $\sigma_n\to+\infty$ we can extract a subsequence $\sigma_{n_k}$ such that
	$$
	\lim_{k\to+\infty} \alpha_{j,\sigma_{n_k}}= \alpha_{j,\infty} \ \ \ \mbox{uniformly on}  \ x\in S^1,
	$$
	for some  $\alpha_{j,\infty}\in C(S^1;\R)$.
	Again by the stability of viscosity solutions we see that
	$\alpha_{j,\infty}$ must solve the limit equation
	$$
	\alpha_{j,\infty}''(x)=0,  \ \ \ \ x\in S^1,
	$$
	hence, it must be $\alpha_{j,\infty}\equiv c_0$ for some $c_0\geq 0$. to find the value of $c_0$ we may integrate  \eqref{ODEalpha4} over $S^1$ getting
	$$
	\displaystyle{\int_{S^1} (\rho_j+\delta(x))\alpha_{j,{\sigma}}(x)dx=\int_{S^1}\widehat{w_j}(x)dx.}
	$$
	Letting $\sigma\to +\infty$ above we get
	$$
	c_0= \frac{\int_{S^1}\widehat{w_j}(x)dx}{\int_{S^1} (\rho_j+\delta(x))}.
	$$
	As this value does not depend on the sequence $\sigma_n$ chosen,
	the claim follows.
\end{proof}

\begin{Proposition}\label{prop:rhodelta-app}
	Let Assumption \ref{ass:data} hold and let $\alpha_{j}^{\rho_j,\delta(\cdot), \widehat w_j(\cdot)}$ be the unique solution to \eqref{ODEalpha} for given $\rho_j,\delta(\cdot), \widehat w_j(\cdot)$. Then $\alpha_j$ is nonincreasing with respect to space homogeneous increments of $\rho_j+\delta(\cdot)$ meaning that
	$$
	h,k\in \R, \ h+k\geq 0   \  \ \Longleftrightarrow \ \ \alpha_{j}^{\rho_j+h,\delta(\cdot)+k,w_j(\cdot)}(x)\leq \alpha_{j}^{\rho_j,\delta(\cdot)w_j(\cdot)}(x) \ \ \ \forall x\in S^1.
	$$
	Moreover, with the same meaning, $\alpha_j$ is nondecreasing with respect to space homogeneous
	increments of $\widehat w_j(\cdot)$.
\end{Proposition}
\begin{proof}
	We start proving the first claim. Considering \eqref{risol}--\eqref{defalpha} we have
	$$\alpha_{j}^{\rho_j,\delta(\cdot),\widehat w_j(\cdot)}(x)- \alpha_{j}^{{\rho_j}+h,{\delta}(\cdot)+k, \widehat w_j(\cdot)}(x)= \int_0^{\infty} \left(1- e^{-t(h+k)}\right) e^{-t \L}\widehat{w_j}\,dt,
	$$
	and the claim follows since $e^{-t \L}\widehat{w_j}$ is a positive operator, i.e. it maps nonnegative functions into nonnegative functions (see, e.g., Section 2 in Chapter II, of Ma and  R\"ockner (1992)).
	
	The second claim follows
	by Corollary \ref{cor:pp}.
\end{proof}
The following proposition establishes the dependence of $\alpha_{j}$ on the territory $M_j$.
\begin{Proposition}
	Let $M_j\subset \widetilde{M_j}\subset S^1$, let $w_j, \widetilde{w_j}$ coefficients associated to $M_j,\widetilde{M_j}$, respectively, and let $\alpha_j$, $\widetilde{\alpha_j}$ be the associated solutions to \eqref{ODEalpha}. Assume that $\widetilde{w}|_{M_j}=w_j$. Then $\alpha_j\leq \widetilde{\alpha_j}$.
\end{Proposition}

\begin{proof}
	It follows from Corollary \ref{cor:pp}.
\end{proof}

\subsection{{Reformulation of the objective functionals}} With the help of the functions $\alpha_{j}$, 
we are able to rewrite the functionals $J_j$ in such a way to get rid of the state equation $P(t)$ in them. This enable us to solve in a simple way for the open loop Nash equilibrium. The result is provided below.
\begin{Proposition}
	\label{pr:oo}
	We have
	\begin{align}\label{eq:Finfdimequiv}
	J_j^{ (I_{-j},B_{-j})}\big(p_0;(I_j,B_j)\big)
	&=\int_0^{\infty} e^{-\rho_j t} \left\langle
	\frac{\big((A_j-1) I_j(t)-B_j(t)\big)^{1-\gamma_j}}{1-\gamma_j},\mathbf{1}_{M_j}\right\rangle_{H_j} dt\nonumber
	\\[3mm]
	&-\int_0^{\infty}e^{-\rho_j  t} \left\langle\alpha_j|_{M_j}, I_j(t)-\eta B_j(t)^{\theta}\right\rangle_{H_j} dt
	\\[3mm]
	&-\langle \alpha_j, p_0\rangle_H
	-\sum_{k=1, \, k\neq j}^N \int_0^{\infty}e^{-\rho_j  t} \left\langle\alpha_j|_{M_k}, I_k(t)-\eta B_k(t)^{\theta}\right\rangle_{H_k} dt.\nonumber
	\end{align}
\end{Proposition}
\begin{proof}
	We only have to rewrite the first and the third term of the expression of 
	$J_j$ provided in \eqref{eq:Finfdim}.
	
	Using \eqref{risol} and \eqref{defalpha}, the first term of
	$J_j$ in \eqref{eq:Finfdim} can be rewritten as follows (recall the definition of $\widehat{w_j}$ in \eqref{eq:defwhatj})
	\begin{align*}
	\left\langle \widehat{w_j}, \int_0^{\infty} e^{-(\rho_j -\L) t}p_0\, dt\right\rangle_H
	=\left\langle\widehat{w_j},(\rho_j -\L)^{-1} p_0\right\rangle_H=\left\langle {(\rho_j -\L^{*})}^{-1}\widehat{w_j}, p_0\right\rangle_H =\left\langle \alpha_j, p_0\right\rangle_H.
	\end{align*}
	In the remainder of the proof, for simplicity of notation, we define the  net emissions
	\begin{equation}\label{eq:Chat}
	N(t):= I(t)-\eta B(t)^{\theta}.
	\end{equation}
	Now, using again \eqref{risol} and \eqref{defalpha}, the third term
	of $J_j$ in \eqref{eq:Finfdim}, can be rewritten by exchanging
	the integrals as follows:
	\begin{equation}\label{ppp}\begin{split}
	&\int_0^{\infty}\left(\int_0^t e^{-\rho_j  t} \left\langle \widehat{w_j}, e^{(t-s)\L} N(s)\right\rangle_Hds\right) dt\\&=
	\int_0^{\infty}\left(\int_0^t e^{-\rho_j s} \left\langle \widehat{w_j}, e^{-(\rho_j-\L)(t-s)} N(s)\right\rangle_Hds\right) dt\\&=
	\int_0^{\infty}e^{-\rho_j  s}
	\left\langle\widehat{w_j},\int_s^{\infty} e^{-(\rho_j-\L) (t-s)} N(s) dt\right\rangle_H ds\\
	=&\int_0^{\infty}e^{-\rho_j  s}
	\left\langle\widehat{w_j},(\rho_j-\L)^{-1} N(s)\right\rangle_H ds\\&
	=\int_0^{\infty}e^{-\rho_j  s} \left\langle{[(\rho_j-\L)^{-1}]^{*}}\widehat{w_j}, N(s)\right\rangle_H ds
	\\&
	=\int_0^{\infty}e^{-\rho_j  s} \left\langle{(\rho_j-\L^{*})^{-1}}\widehat{w_j}, N(s)\right\rangle_H ds
	=\int_0^{\infty}e^{-\rho_j  s} \left\langle\alpha_j, N(s)\right\rangle_H ds.
	\end{split}
	\end{equation}
	Now, using \eqref{eq:Chat}, we get
	\begin{equation}\label{pppbis}\begin{split}
	&\int_0^{\infty}e^{-\rho_j  s} \left\langle\alpha_j, N(s)\right\rangle_H ds
	=\int_0^{\infty}e^{-\rho_j  s} \left\langle\alpha_j, \sum_{k=1}^N (I_k(s)-\eta B_k(s)^{\theta})\right\rangle_H ds
	\\
	&= \int_0^{\infty}e^{-\rho_j  s} \left\langle\alpha_j\mathbf{1}_{M_j}, (I_j(s)-\eta B_j(s)^{\theta})\mathbf{1}_{M_j}\right\rangle_{H_j} ds\\&\ - \sum_{k=1, \, k\neq j}^N \int_0^{\infty}e^{-\rho_j  s} \left\langle\alpha_j\mathbf{1}_{M_k}, (I_k(s)-\eta B_k(s)^{\theta})\mathbf{1}_{M_k}\right\rangle_{H_k} ds.
	\end{split}
	\end{equation}
	The claim easily follows by rearranging the terms.
\end{proof}

\subsection{{Characterization of the solutions: the proof of the existence and uniqueness results of Section \ref{se:noncoopgame} }}
\label{app:proofsect3}

\begin{proof}[Proof of Theorem \ref{th:oo-sect2}]
	{
		Given the additive nature of the expression \eqref{pppquaterbis}, it is clear that the optimization of  it over $(i_{j},b_{j})$ is independent of $(i_{-j},b_{-j})$.
		In this way the original game was decoupled into $N$ independent optimal control problems where the coupling of the original system only enters through the role played by the functions 
		$\alpha_j$'s. Moreover,  the optimization to be performed by agent $j$
		turns out to be just a pointwise optimization of the integrand of the first term in \eqref{pppquaterbis}, i.e. 
		of 
		$$
		\frac{\big((A_j(x)-1) i_j(t,x)-b_j(t,x)\big)^{1- \gamma_j}}{1-\gamma_j} - \alpha_j(x) \left(i_j(t,x)-{\eta}b_j(t,x)^{\theta}\right).
		$$ 
		Fix $(t,x)\in \R_+\times M_j$.  By strict concavity of the integrand function with respect to $i_j(t,x)$ and $b_j(t,x)$, the {unique} maximum point can be found just by
		first order optimality conditions.
		The resulting system is
		\begin{equation}\label{pppquinter}
		\begin{cases}
		\big((A_j(x)-1) i_j(t,x)-b_j(t,x)\big)
		^{- \gamma_j}(A_j(x)-1)- \alpha_j(x)=0,
		\\\\
		-\big((A_j(x)-1) i_j(t,x)-b_j(t,x)\big)
		^{- \gamma_j}
		+\alpha_j(x){{\eta}}\theta
		b_j(t,x)^{\theta-1}=0.
		\end{cases}
		\end{equation}
		Then straightforward computations imply that there exists exactly one open loop equilibrium, whose expression is provided by \eqref{bopt-sect2}-\eqref{iopt-sect2}.	The expressions of the welfare functions follows.
	}
\end{proof}

%
%
%

\begin{proof}[Proof of Theorem \ref{th:uniqueness}]
{We need to show uniqueness.
		Let  $(\phi^{*}_{k})_{k=1,...N}$ be a Markovian equilibrium in the class $\mathcal{D}$. We can identify it with  $(\phi^{*}_{k})_{k=1,...,N}= (\Phi^{*}_{k},\iota^{*}_{k})_{{k=1,...,N}}$.
		This means that, for each $j=1,...,N$, the function  $i_{j}^*(t,x)=( \Phi^{*}_{j}p^{i^{*}_{j}}(t,\cdot))(x)+\iota^{*}_{j}(x)$ is an optimal control of Player $j$ for the  maximization problem 
		\begin{small}
			\begin{align}\label{pppquaterbis2}
			\sup_{i_{j}\in \mathcal{A}_{j}}&\Bigg\{\int_0^{\infty}e^{-\rho_j s}\left(\int_{M_j}
			\left[\frac{\big((A_j(x)-1) i_j(t,x)\big)^{1- \gamma_j}}{1-\gamma_j} - \alpha_j(x) i_j(t,x)
			\right]dx\right) ds\nonumber\\
			&-\int_{S^{1}}p_0(x) \alpha_j(x)dx
			-\sum_{k=1, \, k\neq j}^N \int_0^{\infty}e^{-\rho_j  t}\left( \int_{M_{k}}\alpha_j(x) (\Phi_k^{*}p^{{i_{j}}}(t,\cdot))(x)+\iota^{*}_k(x)) dx \right)dt\Bigg\},\nonumber
			\end{align}
		\end{small}
		where	 $p^{i_{j}}$ solves
		\begin{equation*}\label{eq:stateequation-sect2bis}
		\begin{cases}
		\displaystyle{\frac{\partial p}{\partial t}(t,x) = \sigma \frac{\partial^2 p}{\partial x^2} (t,x)  - \delta(x) p(t, x) +i_{j}(t,x)\mathbf{1}_{M_{j}}(x)}\\ \ \ \ \ \ \ \ \ \ \ \ \ \ \ \displaystyle{ +\sum_{k=1, \, k\neq j}^N  (\Phi^{*}_kp(t,\cdot))(x)+\iota_{k}^*(x))\mathbf{1}_{M_{k}}(x),}\medskip\\
		p(0,x)=p_0(x),  \ \ \ x\in S^1.
		\end{cases}
		\end{equation*}
		With similar arguments as the ones used employed to rewrite the functional of the game (the structure is the same: linear state equation and linearity of the objective functional with respect to the state variable), 
		one rewrites the functional of this optimal control problem and  finds that there exists a unique optimal control for such a problem and that it  is a function constant in time.  By symmetry, this holds for each $k=1,...,N$. Hence, the affine Markovian equilibrium needs to be an open loop one. By uniqueness of the open loop equilibrium, we conclude that  
		$$(\Phi_{k}^{*}p^{\phi^*}(t,\cdot))(x)+\iota^{*}_{k}(x)=\phi_{k}^{*}(x,p^{\phi^*}(t,\cdot)))=   \alpha_k(x)^{-\frac{1}{\gamma_k}}(A_k(x)-1)
		^{\frac{1-\gamma_k}{\gamma_k}}.$$
		In particular $t\mapsto \Phi_{k}^{*}p^{\phi^*}(t,\cdot)(x) $  has to be constant. Since the equilibrium is subgame perfect, we can choose $p_{0}=0$. Then, using Lemma \ref{lemma:phi} with $\Phi=\Phi_{k}^{*}$ and $h(x)=\sum_{k=1}^{N}\alpha_k(x)^{-\frac{1}{\gamma_k}}(A_k(x)-1)^{\frac{1-\gamma_k}{\gamma_k}} \mathbf{1}_{M_{k}}(x)$  enables us to conclude that 
		$\Phi^{*}_{k}=0$, and, consequently, 
		$$\iota^{*}_k(x)=  \alpha_k(x)^{-\frac{1}{\gamma_k}}(A_k(x)-1)
		^{\frac{1-\gamma_k}{\gamma_k}}.$$ 
		The claim follows.
		%
		%
	}
\end{proof}

\subsection{{Some features  of the solution}}
{Setting
	$$
	I^*:=0\cdot \mathbf{1}_{M_0}+\sum_{j=1}^N\mathbf{1}_{M_j} I^*_{j}, \ \ \ B^*:=0\cdot \mathbf{1}_{M_0}+\sum_{j=1}^N \mathbf{1}_{M_j}B^*_j,$$
	where 
	$$
	I_{j}^{*}=i^{*}_{j}(\cdot), \ \ \ B_{j}^{*}=b_{j}^*(\cdot),
	$$
	with $i^*_{j},b^{*}_{j}$ given by \eqref{iopt-sect2}-\eqref{bopt-sect2},
	and  defining the \emph{optimal net emission} as:
	$$N^*:= I^*-(\eta B^*)^{\theta},$$
	the equilibrium state  is
	\begin{equation}\label{eq:p*}
	P^*(t)= e^{t\L}p_0+\int_0^t e^{(t-s)\L} N^*ds.
	\end{equation}
\begin{Corollary}
	\label{cor:Pinfty}
	Let Assumption \ref{ass:data} hold.  Furthermore, assume that  there exists  $\bar{\delta}>0$ such that
	\begin{equation}\label{assf}
	\frac12 v'(x)+\delta(x)\geq \bar{\delta} \ \ \ \forall x\in S^1.\end{equation}
	Then
	$
	\lim_{t\to\infty} P^*(t)= P^*_\infty$ in $H$, where
	$P^*_\infty$ is the unique solution in $H$ to the abstract ODE \ $\L P^*_\infty+N^*=0$.
\end{Corollary}
\begin{proof}
	Let us split
	$\mathcal{L}=\bar{\mathcal{L}}+\bar{\mathcal{D}},$ where
	$$\bar{\L}\varphi:={\L}\varphi-\bar{\delta}\varphi, \ \ \ \varphi\in D(\L); \ \ \ \ \ \bar{\mathcal{D}}\varphi:=-\bar{\delta} \varphi, \ \ \ \varphi\in H.$$
	We can rewrite
	$$P^*(t)=e^{-\bar{\delta} t} e^{t\bar{\L}}p_0+\int_0^t e^{-\bar{\delta}(t-s)}e^{(t-s)\bar{\L}}  N^*ds,
	$$
	and take the limit above when $t\to\infty$. By \eqref{assf}, $\bar{\L}$ is dissipative, hence $e^{s\bar{\L}}$ is a contraction. Therefore, the first term of the right hand side converges to $0$ in $H$, whereas the second one converges in $H$ to
	$$P^*_\infty:=\int_0^\infty e^{-\bar{\delta}s}e^{s\bar{\L}}  N^*ds.$$
	Then, the limit state $P^*_\infty$ can be expressed using  again Engel and Nagel (1995), Proposition 3.14, page 82 and Chapter II, Theorem 1.10, as
$
	P^*_\infty=(\bar{\delta}-\bar{\L})^{-1}  N^*,
$
	i.e.
	$P^*_\infty$ is the solution in $H$ to
	$
	(\bar{\delta}-\bar{\L})P^*_\infty={N}^*,
	$ i.e. to $\L P^*_\infty+{N}^*=0$.
\end{proof}
}
\begin{Remark}\label{rm:deponwanddelta}
	It is clear that, for every $j=1,\dots,N$, the cost functional $J_j$ is decreasing with respect to $\overline{w}_j$; this follows from the fact that a lower cost of pollution makes the welfare bigger, hence also the corresponding welfare function $v_j$ is decreasing in
	$\overline w_j$. 	
	
	Moreover, if, for $i\ne j$, $\overline w_i$ increases, then $J_j$, and so $v_j$ also increase. This can be seen looking at the decomposition of Proposition \ref{pr:oo} or simply observing that the increase in $\overline w_i$ does not modifies the strategy for the $j$-th agent, but makes the $p$ to globally decrease since the agent $i$ will pollute less.
	
	Finally, since $p$ is decreasing with respect to $\delta$ (this comes from the fact that a higher self-cleaning capacity makes the pollution lower), for every $j=1,\dots,N$, the cost functional $J_j$ is decreasing with respect to $\delta$, hence also the corresponding welfare function $v_j$ is decreasing in $\delta$.
\end{Remark}

\begin{Remark}\label{rm:netemissions}
	We now look at the dependence of the
	optimal net emissions
	$n^{*}(t,x):=i_j^*(t,x) - {{\eta}} (b_j^*(t,x))^\theta$, on the data.
	By a simple computation we see that for every $x\in S^1$
	\begin{align}
	&
	i_j^*(x) - {{\eta}} (b_j^*(x))^\theta =
	\alpha_j(x)^{-\frac{1}{\gamma_j}}
	(A_j(x)-1)^{\frac{1-\gamma_j}{\gamma_j}}
	-
	\left({{\eta}}\theta\right)^{\frac{1}{1-\theta}}
	(A_j(x)-1)^{\frac{\theta}{1-\theta}}
	\left[\theta^{-1}-1\right].
	\end{align}
	Then, one can
	analyze the monotonicity of net emissions $n^{*}$ with respect to some parameters. For example:
	\begin{itemize}
		\item[--] when $\gamma_j>1$,  the value of $n^{*}(x)$  is decreasing with respect to
		$A_j(x)$;
		\item[--] when
		$$
		\frac{1-\gamma_j}{\gamma_j}\ge \frac{\theta}{1-\theta}
		\quad and \quad
		\alpha_j(x)^{-\frac{1}{\gamma_j}}\ge
		\left[\theta^{-1}-1\right]
		\left({{\eta}}\theta\right)^{\frac{1}{1-\theta}},
		$$
		the value of $n^{*}(x)$ is increasing with respect to
		$A_j(x)$;
		\item[--] by Proposition \ref{prop:rhodelta-app}, $n^{*}(x)$ is increasing with respect to space homogeneous increments of $\rho_j+\delta(\cdot)$ in the sense specified  in the same proposition.
	\end{itemize}
	
\end{Remark}
\begin{Remark}\label{rm:depPinfty}
	The dependence of $P^{*}_\infty$ on $A_j(\cdot)$ follows from what observed in Remark \ref{rm:netemissions}. Indeed, since $\call$ does not depend on the $A_j(\cdot)$, $P^{*}_\infty$ depends on it only through the stationary optimal net emissions,  $K^*$.
\end{Remark}

\subsubsection{Proof of Proposition \ref{prop:comparison}.}\label{sub:terr}
Just for simplicity of notation, we prove the claim in a very special case, i.e. when $\Pi^1=\{M^1_1,M^1_2\}$  and $\Pi^2=\{M^2_1\}$. The proof of the general claim is a straightforward generalization. Let $n^{1,*}, n^{2,*}$ be the  optimal net emissions associated to $\Pi^1,\Pi^2$, respactively. They differ only for the terms containing
$\alpha$. Now, with clear meaning of the symbols, we have
$$
\alpha^1_1=({\rho-\mathcal{L}^{*}})^{-1}\widehat{w}^1_1, \ \ \  \alpha^1_2=({\rho-\mathcal{L}^{*}})^{-1}\widehat{w}^1_2, \  \ \ \ \alpha^2_1={(\rho-\mathcal{L}^{*})}^{-1}\widehat{w}^2_1.
$$
Since $\widehat{w}^2_1\geq \widehat{w}^1_1$ and $\widehat{w}^2_1\geq \widehat{w}^2_1$, by Corollary \ref{cor:pp} it follows that
$$
\alpha_1^2(x)\geq  \alpha_1^1(x) \ \ \forall x\in M_1^1, \ \ \  \alpha_1^2(x)\geq  \alpha_2^1(x) \ \ \forall x\in M_2^1.
$$
Then the claim follows from \eqref{netemission}, \eqref{bopt-sect2}, \eqref{iopt-sect2}, \eqref{eq:Ymild}, and since the operator  $e^{t\mathcal{L}}$ is positive preserving, i.e. maps nonnegative functions into nonnegative ones (see, e.g.,  Ma and  R\"ockner (1992), Ch\,. II, Sec.\,2).
\hfill$\square$

		\begin{Proposition}
			\label{cor:duebanchmark}
			Let Assumptions \ref{ass:data} hold and assume  that
			$\delta(\cdot)\equiv \delta^o>0$,  $w(\cdot)\equiv w^o>0$, and $\rho_j=\rho$ for all $j=1,...,N$.
			Then
			$$\underline{\alpha}\equiv \frac{w^o}{\rho-\delta^o},
			$$
			which is (constant and) independent of ${\sigma}$; moreover,
			if ${\sigma}=0$(\,\footnote{To be precise, the case ${\sigma}=0$ should be treated separately, as it is out of our assumptions. 
		Nonetheless, this case can be easily treated pointwiseless on $x$ and gives raise, in the case under consideration  here, to the solutions we illustrate below.}), then
			$$
			\alpha_j\equiv \frac{w^o}{\rho-\delta^o}, \ \ \ \forall j=1,...,N,
			$$
			i.e. the same solution of the cooperative game case obtained for each diffusion coefficient ${\sigma}\geq 0$.
		\end{Proposition}
\begin{proof}
{All the claims can be just checked by plugging the given values into the ODE's for $\underline{\alpha}$ and $\alpha_j$.}
\end{proof}

\subsection{Series expansion of the $\alpha_j$'s}
\label{secapp:Fouries}
In this section
$A_j\equiv A_j^o>1$, $w_j\equiv w_j^o> 0$, $\delta\equiv \delta^o\ge 0$, $\sigma \equiv {\sigma}>0$, $v\equiv v^o\in \R$. We use the identification $S^1\cong 2\pi\R/\Z$ and assume, without loss of generality that $M_j=(0,\ell_j)$. Finally, to save notation, we suppress the subscript $j$. We are going to study the Fourier series expansion of $\alpha=\alpha_j$ in the case without or with advection. We notice that the convergence of the series is uniform on $S^1$ due to the smoothness of $\alpha$.

\subsubsection{The case without advection: $v^o=0$}
In this case $\mathcal{L}=\mathcal{L}^{*}$, hence
 $\alpha$ solves the equation
$
(\rho-\mathcal{L})\alpha=\widehat{w},
$
with $\widehat{w}$ defined as
$$
\widehat{w}(x):=\begin{cases} w^o, \ \ \ \ \ \ \ \ \mbox{if} \ x\in {M},\\
0, \ \ \ \ \ \ \ \ \ \  \mbox{if} \ x\notin M.
\end{cases}
$$
The set of elements of $H$
\begin{small}
\begin{equation}\label{eq:eigenvector}
\left\{\mathbf{e}_0(x):=\frac{1}{\sqrt {2\pi}}\mathbf{1}_{S^1}(x)\right\}\bigcup  \left\{\mathbf{e}_n^{(1)}(x):=\frac{1}{\sqrt \pi} \sin \left(nx\right), \ \mathbf{e}_n^{(2)}(x):=\frac{1}{\sqrt \pi} \cos \left(nx\right), \ n\in\N\setminus\{0\}\right\}
\end{equation}
\end{small}
is an orthonormal basis on $H$.
They are also eigenfunctions of {$\mathcal{L}=\mathcal{L}^{*}$} with associated eigenvalues
\begin{equation}\label{lambdaconst}
\mu_n=-\delta^o-{{\sigma}}n^2, \ \ \ \ n\in\N.
\end{equation}
We can expand in Fourier series
\begin{equation}\label{Fexp}
\alpha=\left\langle\alpha,\mathbf{e}_0\right\rangle_H \mathbf{e}_0+\sum_{n\in\mathbb{N}\setminus\{0\}, \, i=1,2}\langle\alpha,\mathbf{e}_n^{(i)}\rangle_H \,\mathbf{e}_n^{(i)}.
\end{equation}
Let us compute the coefficients of the series. As for $n=0$,
we notice that
$$\langle \alpha,(\rho-\mathcal{L})\mathbf{e}_0\rangle_H=
\langle (\rho- \mathcal{L})\alpha,\mathbf{e}_0\rangle_H
=\langle\widehat{w},\mathbf{e}_0\rangle_H,
$$
hence,
\begin{equation}\label{e0}
\left\langle\alpha,\mathbf{e}_0\right\rangle_H
=\left(\rho+ \delta^o\right)^{-1}
\left\langle\widehat{w},\mathbf{e}_0\right\rangle_H= \left(\rho+ \delta^o\right)^{-1}\frac{1}{\sqrt{2\pi}}\int_0^{2\pi}\widehat{w}(x) dx =\frac{1}{\sqrt{2\pi}}\frac{\ell w^o}{\rho+ \delta^o}.
\end{equation}
Similarly,
$$ \langle\alpha,\mathbf{e}^{(i)}_n\rangle_H
=\left(\rho+ \delta^o+{{\sigma}}{n^2}\right)^{-1}
\langle\widehat{w},\mathbf{e}^{(i)}_n\rangle_H, \ \ \ \ \ \ \ \forall i=1,2, \ \forall n\in\mathbb{N}\setminus\{0\}.
$$
We may compute
\begin{small}
\begin{equation}\label{scalar}
\displaystyle{\langle\widehat{w},\mathbf{e}_n^{(i)}\rangle_H=
	\begin{cases}
	\frac{1}{\sqrt{\pi}}\int_0^{2\pi}\widehat{w}(x)\sin \left({n}x\right) dx=\frac{w^o}{\sqrt{\pi}}\int_0^{\ell}\sin \left({n}x\right) dx=\frac{w^o}{\sqrt\pi} \frac{1}{n}\left[1-\cos \left(n\ell\right)\right],\ \ \ \ \mbox{if} \ i=1,\\\\
	\frac{1}{\sqrt{\pi}}\int_0^{2\pi}\widehat{w}(x)\cos \left({n}x\right) dx=\frac{w^o}{\sqrt{\pi}}\int_0^{\ell}\cos \left({n}x\right) dx=\frac{w^o}{\sqrt\pi} \frac{1}{n}\sin \left(n\ell\right),\ \ \ \ \ \ \ \ \ \ \, \mbox{if} \ i=2.
	\end{cases}}
\end{equation}
\end{small}
Plugging these results into \eqref{Fexp} yields
\begin{align*}
\alpha(x)&=\frac{1}{{2\pi}}\frac{\ell w^o}{\rho+ \delta^o}  +\frac{{w^o}}{\pi}  \sum_{n=1}^\infty  \frac{\sin \left(nx\right)\left(1-\cos \left({n}\ell\right)\right)+\sin (n\ell)\cos\left(nx\right)}{n\left(\rho+ \delta^o+{{\sigma}}{n^2}\right)},
\end{align*}
\subsubsection{The case with advection: $v^{o}\neq 0$}
Recalling the expression of $\mathcal{L}^\star$ provided in \eqref{eq:L*}, we have in the present case
$$[\mathcal{L}^\star\psi](x)= {\sigma}  \psi''(x)-v^o\psi'(x){-\delta^o}\psi(x),
\ \ \  \psi\in D(\mathcal{L^\star}).
$$
Consider again the family \eqref{eq:eigenvector}.
In this case $\mathbf{e}_0$ is still an eigenfunction of $\mathcal{L}^\star$, but  $\mathbf{e}^{(i)}_n$ are not eigenfunction of $\mathcal{L}^\star$ anymore for $n\in \mathbb{N}\setminus\{0\}$. However, still $H_n:=\mbox{Span}\{\mathbf{e}^{(1)}_n, \mathbf{e}^{(2)}_n\}$ are invariant subspaces for $\mathcal{L}^\star$ for each $n\in\mathbb{N}\setminus \{0\}$.
Indeed
$$
\mathcal{L}^\star\mathbf{e}^{(1)}_n= -{\sigma} n^2\mathbf{e}^{(1)}_n-v^on\mathbf{e}^{(2)}_n{-\delta^o} \mathbf{e}^{(1)}_n; \ \ \  \mathcal{L}^\star\mathbf{e}^{(2)}_n= -{\sigma} n^2\mathbf{e}^{(2)}_n+v^on\mathbf{e}^{(1)}_n{-\delta^o} \mathbf{e}^{(2)}_n;
$$
Arguing in a similar way as in the previous subsection we get, for each $n\in\mathbb{N}\setminus\{0\}$, the couple of equations
$$
(\rho+{\sigma}n^2+\delta^o)\langle \alpha, \mathbf{e}^{(1)}_n\rangle_H +v^on\langle \alpha, \mathbf{e}^{(2)}_n\rangle_H= \langle \widehat w, \mathbf{e}^{(1)}_n\rangle_H,
$$
$$
(\rho+{\sigma}n^2+\delta^o)\langle \alpha, \mathbf{e}^{(2)}_n\rangle_H -v^on\langle \alpha, \mathbf{e}^{(1)}_n\rangle_H= \langle \widehat w, \mathbf{e}^{(2)}_n\rangle_H.
$$
yielding, for each $n\in\mathbb{N}\setminus\{0\}$,
$$\langle \alpha, \mathbf{e}^{(1)}_n\rangle_H
= \frac{(\rho+{\sigma}n^2+\delta^o)\langle \widehat w, \mathbf{e}^{(1)}_n\rangle_H-v^on \langle \widehat w, \mathbf{e}^{(2)}_n\rangle_H}{(\rho+{\sigma}n^2+\delta^o)^2+(v^on)^2},$$
$$\langle \alpha, \mathbf{e}^{(2)}_n\rangle_H
= \frac{(\rho+{\sigma}n^2+\delta^o)\langle \widehat w, \mathbf{e}^{(2)}_n\rangle_H+v^on \langle \widehat w, \mathbf{e}^{(1)}_n\rangle_H}{(\rho+{\sigma}n^2+\delta^o)^2+(v^on)^2}.
$$
Using \eqref{Fexp}--\eqref{scalar} and the expressions above, we have for $x\in S^1$
\begin{equation*}
\begin{split}
\alpha(x)&= \frac{1}{{2\pi}}\frac{\ell w^o}{\rho+ \delta^o}+\frac{w^o}{ \pi}\sum_{n=1}^\infty  \frac{(\rho+{\sigma}n^2+\delta^o) \sin (n\ell)+v^on(1-\cos(n\ell)) }{n
	\left(\rho+\delta^o+{\sigma}n^2\right)}\cos (nx)\\&
+ \frac{w^o}{ \pi}\sum_{n=1}^\infty  \frac{(\rho+{\sigma}n^2+\delta^o) (1-\cos (n\ell))-v^on \sin (n\ell) }{n
	\left(\rho+\delta^o+{\sigma}n^2\right)}\sin (nx).
\end{split}
\end{equation*}

\end{footnotesize}

\end{document}